%% file: ddcontrol_poisoning_extended.tex
\documentclass[letterpaper, 10 pt, conference]{ieeeconf}  

\IEEEoverridecommandlockouts                              

\usepackage{cite}
\usepackage{amsmath}
\usepackage{amssymb,amsfonts}
\usepackage{graphicx}
\usepackage{grffile}
\usepackage{textcomp}
\usepackage{url}
\usepackage{float}
 \usepackage[linesnumbered, ruled, vlined]{algorithm2e}
 \SetKwRepeat{Do}{do}{while}%
 \usepackage{bbm}
 
 \usepackage{url}

\usepackage{color, colortbl}
\definecolor{Gray}{gray}{0.95}
\definecolor{LightCyan}{rgb}{0.88,1,1}
\newcommand{\isextended}{extended}

\DeclareMathOperator*{\argmin}{arg\,min}
\usepackage{bbm}
\usepackage{mathtools}
\usepackage{xcolor}
\usepackage{subcaption}
\newtheorem{assumption}{Assumption}
\newtheorem{lemma}{Lemma}
\newtheorem{proposition}{Proposition}
\newcommand{\set}{\mathcal}

\DeclareMathOperator{\Tr}{Tr}
\DeclareMathOperator{\vect}{vec}

\DeclareMathOperator{\supp}{supp}
\DeclareMathOperator{\rank}{rank}

\DeclareMathOperator{\sigmamax}{\sigma_{\text{max}}}
\DeclareMathOperator{\sigmamin}{\sigma_{\text{min}}}

\newtheorem{definition}{Definition}
\newtheorem{corollary}{Corollary}
\newtheorem{remark}{Remark}
\let\oldref\ref
\renewcommand{\ref}[1]{(\oldref{#1})}

\SetKwComment{Comment}{$\triangleright$\ }{}
\SetCommentSty{itshape}
\title{\LARGE \bf
Poisoning Attacks against Data-Driven Control Methods\\(Extended version)
}

\author{Alessio Russo$^{1,^\star}$ and Alexandre Proutiere$^{1}$
\thanks{$^\star$ Corresponding Author.}%
\thanks{$^{1}$Alessio Russo and Alexandre Proutiere are in the Division of Decision and Control Systems of the EECS School at KTH Royal Institute of Technology, Stockholm, Sweden.
        {\tt\small \{alessior,alepro\}@kth.se}}

}

\begin{document}

\maketitle
\thispagestyle{empty}
\pagestyle{empty}
\tableofcontents


\input{1.introduction.tex}
\input{2.model.tex}
\input{3.attack.tex}
\input{4.simulations.tex}

\input{5.conclusion.tex}
\section*{APPENDIX}
Here in the appendix we present some additional result.  First we show the targeted attack criterion for VRFT, and then an attack on a control law designed by means of Willems lemma.
\input{6.appendix_vrft_targeted.tex}

\input{7.appendix_willems.tex}

\input{8.appendix_additional_stuff.tex}

\bibliographystyle{IEEEtran}
\bibliography{IEEEabrv,ref}

\end{document}

%% file: 1.introduction.tex
\begin{abstract}
	This paper investigates poisoning attacks against data-driven control methods. This work is motivated by recent trends showing that, in supervised learning, slightly modifying the data in a malicious manner can drastically deteriorate the prediction ability of the trained model. We extend these analyses to the case of data-driven control methods. Specifically, we investigate how a malicious adversary can poison the data so as to minimize the performance of a controller trained using this data. We show that identifying the most impactful attack boils down to solving a bi-level non-convex optimization problem, and provide theoretical insights on the attack. We present a generic algorithm finding a local optimum of this problem and illustrate our analysis in the case of a model-reference based approach, the Virtual Reference Feedback Tuning technique, and on data-driven methods based on Willems et al. lemma. Numerical experiments reveal that minimal but well-crafted changes in the dataset are sufficient to deteriorate the performance of data-driven control methods significantly, and even make the closed-loop system unstable.
\end{abstract}

\section{Introduction}Data-driven control tuning has been an active area of research for the last few decades \cite{campi2002virtual,formentin2012non,formentin2014comparison,hjalmarsson1998iterative,karimi2004iterative,karimi2007non,sala2005extensions,campestrini2011virtual,lequin2003iterative,de2019formulas,coulson2019data}: it allows the user to perform data-to-controller design, and to avoid identifying a specific model for the plant. Data-driven methods have been used in many systems for which it is difficult to derive, from first-principle, a mathematical model of the plant. This is especially the case for multi-input/multi-output (MIMO) linear time-invariant systems, where it may be hard or too costly to model the interaction between the inputs and the outputs of the system \cite{formentin2012non}.

Unlike adaptive control, where the controller is usually updated in an \textit{online} manner, data-driven methods tune the control policy in an \textit{offline} way.
Specifically, the controller is computed upon a batch of data and an underlying objective function.
]Several data-driven control techniques have been proposed in the literature: iterative methods, such as iterative feedback tuning \cite{hjalmarsson1998iterative} and correlation-based tuning \cite{karimi2004iterative}; one-shot methods such as Virtual Reference Feedback Tuning (VRFT) \cite{campi2002virtual}, the correlation approach \cite{karimi2007non}, and  recent techniques \cite{de2019formulas,coulson2019data} leveraging Willems et al. lemma \cite{willems2005note}.

Data-driven control methods draw similarities to other machine learning techniques. As in supervised and unsupervised learning, these methods make use of a batch of data to train upon. Recently, there has been a surge of interest in studying how a malicious agent can deteriorate the performance of supervised learning methods \cite{goodfellow2014explaining,biggio1}, and more recently, reinforcement learning methods too \cite{russo2019optimal}. It has been shown that through barely perceptible but specific changes in the dataset, namely \textit{data-poisoning attacks}, the malicious agent is capable of reducing the classifier performance by a significant amount \cite{biggio1,jagielski2018manipulating,munoz2017towards}. If attacks can significantly reduce the classification performance by slightly altering the dataset available to the user, it makes it difficult to assess the integrity of the data. 

In this paper, we examine poisoning attacks against data-driven control methods. These methods follow the paradigm of \textit{learning from the data}, and we show that as such, they suffer from the same problem as classical machine learning methods. Data-driven control methods use experimental data to design control laws directly, and we consider an attacker who may affect the data being recorded during the experiment or maliciously change it after the experiment is done. Over the last decade, researchers have developed a risk management framework for cyber-physical systems \cite{chong2019tutorial,sandberg2015cyberphysical}, comprising \textit{detection, prevention} and \textit{treatment} of attacks. To the best of our knowledge, this paper is the first to analyze data-poisoning attacks on data-driven control methods. Our contributions are as follows.\\ 

(i) We formalize the problem of devising efficient poisoning attacks with a bounded amplitude as a bi-level optimization problem (whose objective depends on the attacker's goal, e.g., minimize the learner's performance). In general, this problem is non-convex. We develop gradient-based algorithms to find local optima of the problem and compute poisoning attacks.\\

(ii) We specify our algorithms to the case of VRFT. There, the use of parametrized controller eases the computation of gradients, and we can theoretically quantify the potential impact of poisoning attacks. We also investigate poisoning attacks for data-driven methods based on Willems et al. lemma\cite{willems2005note,de2019formulas}\ifdefined\isextended\else, but due to space constraints, this analysis is presented in the associated technical report  \cite{techreport}\fi.\\

(iii) We experimentally assess the impact of poisoning attacks. It turns out that signals with larger levels of excitation (so as to identify the optimal controller more easily in data-driven methods) may make attacks more efficient (in the sense that the resulting closed-loop system is unstable). Our experiments suggest that poisoning attacks, even with very low amplitude, can significantly deteriorate the efficiency of data-driven control methods. Minimal changes in the dataset can cause system instability.

%

%% file: 2.model.tex
\section{Model and problem formulation}
\subsection{Notation}
We consider discrete-time models, indexed by $t\in \mathbb{N}_0$, and we will indicate by $[N]$ the sequence of integers from $0$ to $N$. We denote by $z$ the one-step forward shift operator and by $\mathcal{H}_2$ the Hardy space of complex functions which are analytic in $|z|< 1$ for $z\in \mathbb{C}$. For a generic signal $x_t \in \mathbb{R}^{n}$, we denote its values in $t\in[N]$ as \[X_{[N]}=\begin{bmatrix}
x_0 & x_1 &\dots & x_{N}
\end{bmatrix}^\top\in \mathbb{R}^{N\times n}\]
and, in case it is clear from the context, we omit the subscript notation. The vectorized version of $X_{[N]}$ is denoted by $\boldsymbol{x}_{[N]}=\vect(X_{[N]})$ (if $l=1$ then $\boldsymbol{x}_{[N]} = X_{[N]})$. For a vector $x\in \mathbb{R}^n$, we denote by $\supp(x)=\{i: x_i\neq 0\}$ the set of non-zero entries of $x$. Similarly, we indicate the $0$-`norm` of $x$ by $\|x\|_0=|\supp(x)|$. For a generic matrix $A$ denote its $i$-th row by $A_{i,*}$, and let $\lambda_i(A)$ be the $i$-th eigenvalue of $A$ (in case $A$ is squared). Also, let $\sigmamax(A),\sigmamin(A)$ be respectively the maximum and minimum singular value of $A$.  Finally, for a vector $x\in \mathbb{R}^n$ and a function $f: \mathbb{R}^n\to \mathbb{R}$, we denote by $\nabla_x f(x)$ the $n$-dimensional vector of partial derivatives, where each element is $\partial_{x_i} f(x)$ with $\partial_{x_i} = \frac{\partial }{\partial x_i}$.

\subsection{Plant dynamics}
Throughout the paper, we consider a controllable and observable discrete time linear single-input/single-output system of the form
\begin{align}\label{eq:system_state_space1}
x_{t+1} &= Ax_t +Bu_t,\quad x_0\in \set X\\
y_t&=Cx_t+Du_t, \label{eq:system_state_space2}
\end{align}
where $x\in \mathbb{R}^n, u,y\in \mathbb{R}$ and $\set X$ is a closed-convex subset of $ \mathbb{R}^n$. Depending on the algorithm being used, we equivalently use transfer function notation and denote the input-output relationship using transfer function notation $y_t = G(z)u_t,$ with $G(z) = C(zI-A)^{-1}B+D$.  We also denote the multiplication of two transfer functions $G(z)$ and $L(z)$  by $GL(z)$ (similarly the sum). Noise is not included in the model, but it is well known that data-driven control methods can be applied also to noisy data, using for example instrumental variables \cite{campi2002virtual}. The procedure and the concepts described in this work can be directly used to extend the method to the noisy case.
\ifdefined\isextended 
\subsection{Persistency of excitation}
One of the methods we study in the paper makes use of the least squares procedure. This method has been widely studied in literature, we will state some useful definitions:\\
\begin{definition}
A bounded locally square integrable vector function $\phi:\mathbb{R}\to\mathbb{R}^d$ is said to be \textit{persistently exciting} if there exists a constant $t_0$  and a positive constant $T_0$ and $\alpha>0$ such that
\end{definition}
\begin{equation}
\frac{1}{T_0} \sum_{k=t}^{t+T_0} \phi_k\phi_k^\top \succeq \alpha I \succ 0,\quad \forall t\geq t_0.
\end{equation}
The concept of PE function is equivalent to the exponential asymptotic stability of the zero solution of the differential equation $\dot e(t) = -\phi(t)\phi(t)^\top e(t)$ (where $e(t)$ usually denotes the error and $\phi(t)\phi(t)^\top$ is a measure of covariance). The PE condition assures the minimum eigenvalue of $\Phi_N= \sum_{t=0}^{N-1} \phi_t \phi_t^\top$ is strictly positive, which in turns implies that $\Phi_N^{-1}$ exists and the estimate $\hat \theta_N$ is unique. In particular, we have that the eigenvalues of $\Phi_N$ grow at-least linearly in $N$, and the following result from \cite{mareels1988persistency} \[	\liminf_{N	\to\infty} \lambda_{\text{min}}\left(\frac{1}{N} \sum_{t=t_0}^{t_0+N} \phi_t\phi_t^\top\right) \succeq \alpha/2\succ 0\]
which implies that the eigenvalues of the correlation matrix $R_N=\frac{1}{N}\Phi_N$ are lower bounded by $\alpha/2$ for $N\to\infty$. 
\fi
\subsection{Data-driven control}
We assume that the learner has no knowledge of the tuple $(A,B,C,D)$, and, that she uses data-driven methods to find an appropriate control law that satisfies some design requirements.
 We denote by $\set D_N =\{(u_t, y_t), t\in [N-1]\}=(\boldsymbol{u}_{[N-1]}, \boldsymbol{y}_{[N-1]})$ the data available to the learner. This data comes from open-loop experiments on the plant.
 
Various data-driven approach have been proposed in the literature over the last decades: model-reference methods, such as VRFT \cite{campi2002virtual} and correlation-based \cite{karimi2007non}; and more recent methods based on Willems et al. lemma \cite{de2019formulas,coulson2019data}. For simplicity of exposition, we will consider a method based on model-reference, and a method based on Willems lemma. Specifically, we will focus on VRFT\cite{campi2002virtual} and the data-driven approach developed in \cite{de2019formulas}.
 
\subsubsection{Virtual reference feedback tuning} In this method, the design requirements are encapsulated into a reference model $M_r(z)$  that captures the desired closed-loop behavior from $r_t$ to $y_t$, where $r_t \in \mathbb{R}^p$ is the reference signal. We assume that $M_r$ satisfies some realizability assumptions, such as being a proper stable transfer function.
 
The  objective of the model-reference control problem is to find a controller $K_\theta(z)$, parametrized by $\theta\in \mathbb{R}^{n_k}$, for which the closed-loop transfer function matches the user defined reference model $M_r(z)$. In other words, we wish to find the parameter  $\theta$ that minimizes the  $\set H_2$ norm of the difference between the reference model and the closed-loop system $\Delta_\theta(z) = M_r(z)-[(I+GK_\theta)^{-1}GK_\theta](z)$:
 \begin{align}\label{eq:criterion_J_MR}
J_{\text{MR}}(\theta) &= \left\|M_r(z)-[(I+GK_\theta)^{-1}GK_\theta](z)\right\|_2^2\\
&= \frac{1}{2\pi} \int_{-\pi}^{\pi} \Tr\left[\Delta_\theta(e^{j\omega}) \Delta_\theta^\top(e^{-j\omega})\right]\textrm{d}\omega.
\end{align}
It is known in the literature that $J_{\text{MR}}(\theta)$ is non-convex in $\theta$, and thus, difficult to optimize. A first simplification, commonly used \cite{campi2002virtual,karimi2007non,formentin2014comparison}, is to make the following assumption.
\begin{assumption}[\cite{campi2002virtual}]\label{assumption1}
The sensitivity function $I-M_r(z)$ is close to the actual sensitivity function $(I+GK_{\hat{\theta}})^{-1}(z)$ in the minimizer $\hat{\theta}$ of \ref{eq:criterion_J_MR}.
\end{assumption}

This allows us to replace the cost $J_{\text{MR}}(\theta)$ by 
\begin{equation}\label{eq:criterion_J}
J(\theta) = \left\|M_r(z)-[(I-M_r)GK_\theta](z)\right\|_2^2.\\
\end{equation}
\ifdefined\isextended
from which follows that the ideal controller $K^\star(z)$ can be defined through $G(z)$ and $M_r(z)$ as \[K^\star(z) = [G^{-1}(I-M_r)^{-1}M_r](z).\]
\begin{remark}\textit{This controller is usually of high order, and non-causal. Furthermore, it may not belong to the class of control laws to which $K_\theta(z)$ belongs to.\\}\end{remark}\fi
To minimize \ref{eq:criterion_J} without identifying the plant, we can resort to minimizing the difference between the input signal $u_t$ from the experiments and the control signal $  K_\theta(z) e_t$ computed using the \textit{virtual error signal}, $e_t$. The latter is defined as $ e_t =  r_t -y_t = (M_r^{-1}(z)-1)y_t$ where $ r_t$ is the \textit{virtual reference signal} computed using the reference model $M_r(z)$ as $ r_t = M_r^{-1}(z)y_t$. 

Minimizing the squared difference between $u_t$ and $ K_\theta(z) e_t$ usually gives a biased estimate of the minimizer $\theta$ of $J$ (that is the case if the controller that leads the cost function to zero is not in the controller set). To circumvent this issue, it is common to introduce a filter $L(z)$ that will pre-filter the data \cite{campi2002virtual}. We can then  define the objective criterion that is actually minimized:
\begin{equation}\label{eq:criterion_J_VR_N}J_{\text{VR}}^{N}(\theta) = \frac{1}{N} \sum_{t=0}^{N-1} \|u_t - K_\theta(z)  e_t\|_2^2\end{equation}
and it can be proven \cite{campi2002virtual} that for stationary and ergodic signals $\{y_t\}$ and $\{u_t\}$ we  get the following asymptotic result 
$\lim_{N\to\infty}J_{\text{VR}}^N(\theta) =J_{\text{VR}}(\theta)$, where:
\begin{align*}J_{\text{VR}}(\theta) &=  \frac{1}{2\pi} \int_{-\pi}^{\pi} \Tr\left[\bar \Delta_\theta(e^{j\omega}) \Phi_{u}(\omega)\bar\Delta_\theta^\top(e^{-j\omega})\right]\textrm{d}\omega\\
\bar \Delta_\theta(z) &\coloneqq I-[K_\theta(I-M_r)G](z),
\end{align*}
with $\Phi_{u}$ being the power spectral density of $u_t$. Let $K^\star$ denote the minimizer over all possible $K$ of $\left\|M_r(z)-[(I-M_r)GK](z)\right\|_2^2$. If $\in K^\star\in \{K_\theta: \theta\}$, then $K^\star$ is also the minimizer of \ref{eq:criterion_J_VR_N}. Otherwise, one can properly choose a filter $L(z)$ to filter the experimental data so that the minimizer of \ref{eq:criterion_J_VR_N} and \ref{eq:criterion_J} still  coincide (refer to \cite{campi2002virtual} for details).

It is worth mentioning that in practice it is assumed that the control law can be linearly parametrized in terms of a basis of transfer functions:
\begin{assumption}\label{assumption2vrft_linear_control}
	The control law $K$ is represented by an LTI system $K_\theta(z)$ that is linearly parametrized in $\theta \in \mathbb{R}^{n_k}$, and we will write $K_\theta(z) = \beta^\top(z)\theta$, with $\beta(z)$ being a vector of linear discrete-time transfer functions of dimension $n_k$.
\end{assumption}

Assumption \oldref{assumption2vrft_linear_control} includes different types of control law, such as PID, and can be relaxed to other types of parametrization (see \cite{sala2005extensions} and \cite{esparza2011neural}).
\\
\subsubsection{Data-driven control based on Willems lemma}
Recently, there has been a surge of interest in data-driven methods that rely on a lemma by Willems et al \cite{willems2005note}, see for example \cite{de2019formulas,coulson2019data}. In \cite{de2019formulas} Theorem 1,  the authors show that for the system $x_{t+1}=Ax_t+Bu_t$, we can equivalently write
\[ \begin{bmatrix}
B & A
\end{bmatrix}= X_{[1,N]}\begin{bmatrix}
U_{[N-1]}\\X_{[N-1]}
\end{bmatrix}^{\dagger}\]
where $X$ and $U$ represent data collected from the system and $\dagger$ denotes the right inverse. Observe that the above representation holds only if the input sequence is an exciting input of order $n+1$ and $\textrm{rank}\begin{bmatrix}U_{[N-1]} & X_{[N-1]}\end{bmatrix}=n+m$ \ifdefined\isextended\else (please refer to \cite{techreport} for a definition of persistently exciting signal)\fi. Furthermore, as shown in Theorem 2 in \cite{de2019formulas}, any closed-loop system with a state-feedback control $u_t=Kx_t$ we have the following equivalent representation:
\begin{equation}x_{t+1}=X_{[1,N]}^\top G_Kx_t\end{equation}
where $G_K$ is a $T\times n$ matrix satisfying
\begin{equation}K=U_{[N-1]}^\top G_K \hbox{ and } I_n=X_{[N-1]}^\top  G_K.\end{equation}
We can equivalently write  $A+BK=X_{[1,N]}^\top G_K$: this allows to treat $G_K$ as a decision variable, and search for a matrix $G_K$ that satisfies some  performance conditions \cite{de2019formulas}. Due to space constraints, our analysis for attacking this data-driven approach is reported in \ifdefined\isextended the appendix\else\cite{techreport}\fi.


%% file: 3.attack.tex
\section{Poisoning attacks on Data-Driven control}
In this section we introduce a generic framework that can be used to compute attacks for the different of data-driven methods.

\subsection{Attack framework}
\textbf{Attack description and assumptions.} We assume the goal of the malicious agent is to degrade performance of the closed-loop system by minutely corrupting the available data.
\ifdefined\isextended
The attack is thus considered a poisoning attack, and since the attacker has access to some data of the system, confidentiality is breached. Furthermore, this attack also affects both the \textit{integrity} and \textit{availability} of the data, two of the three fundamental properties in computer security \cite{bishop2002art}, together with confidentiality, and as such, may bring severe damage.
\\\\\fi
As in classical data poisoning analysis \cite{biggio1}, we assume the malicious agent knows the optimization problem being solved by the learner: the latter wishes to find a control law $K$ minimizing a cost $\set L(\set D_N, K)$, also denoted by $\set L(\boldsymbol{u},\boldsymbol{y}, K)$ when $\set D_N=(\boldsymbol{u},\boldsymbol{y})$ (see the previous section for examples). We further assume the malicious agent has no knowledge of the plant, and that she can access the data available to the learner. 
\ifdefined\isextended The malicious agent aims to poison the dataset $\set D_N$  to reduce the closed-loop performance of the system, but, at the same time, make sure the poisoned data does not differ too much from the original data in order to remain stealthy (this last assumption can be easily relaxed).\fi The attack framework is illustrated in Figure \ref{fig:poisoning_scheme}. Notice that the malicious agent affects only the data collected from the plant.
\begin{figure}[t]
\centering
\includegraphics[width=0.85\linewidth]{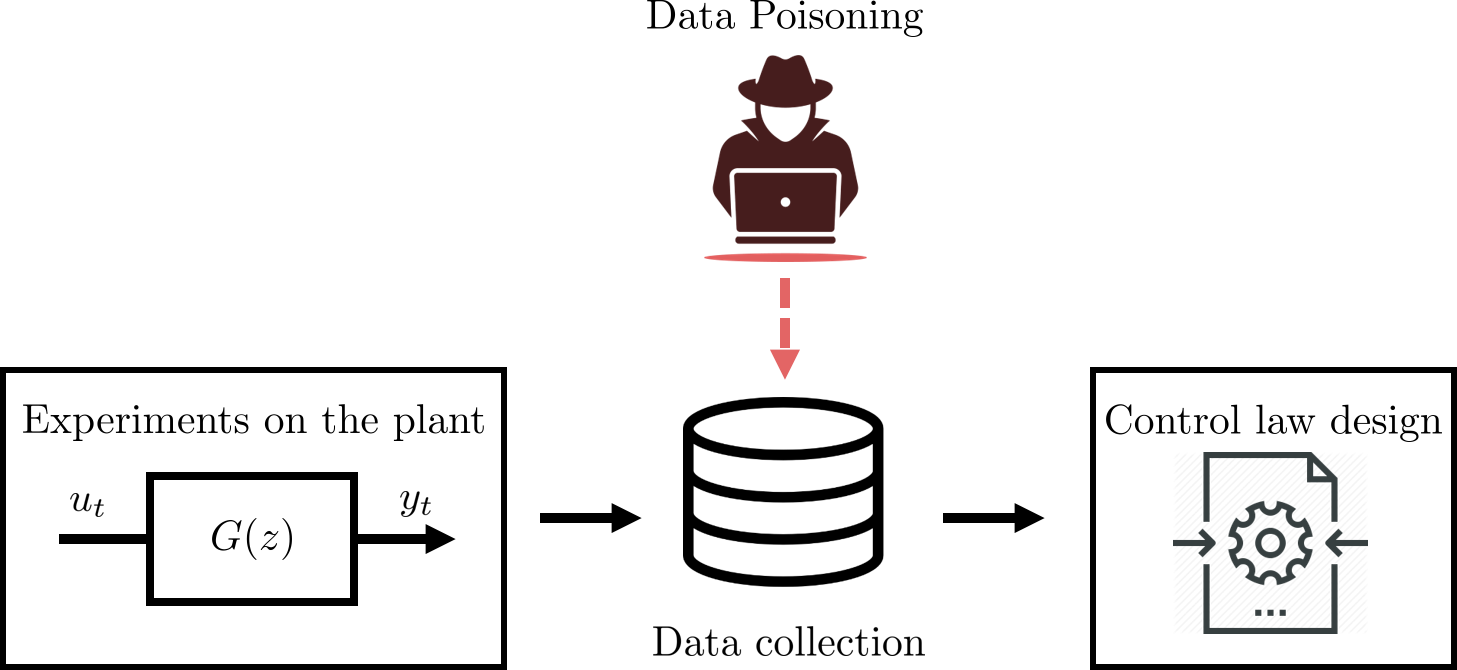}
\caption{Data poisoning scheme for data-driven methods.}
\label{fig:poisoning_scheme}
\end{figure}
We denote the malicious signal on the actuators by $\boldsymbol{a}_{u}\in \mathbb{R}^{N}$, and respectively $\boldsymbol{a}_{y}\in\mathbb{R}^{N}$, the attack signals on the sensors (note that we omitted the time-window subscript). We indicate the attack signal at time $t \in [N-1]$ by $a_{u,t}$ and  $a_{y,t}$. The new input and output data points in the dataset at time $t$ are respectively: $u_t'=u_t+a_{u,t}$  and   $y_t'=y_t+a_{y,t}$. Similarly, we indicate the attacked vector signals by $
\boldsymbol{u}'=\boldsymbol{u}+\boldsymbol{a}_u=\begin{bmatrix}u_0' & u_1' &\dots & u_{N-1}'\end{bmatrix}^{\top}$ and  $\boldsymbol{y}'=\boldsymbol{y}+\boldsymbol{a}_y=\begin{bmatrix}y_0' & y_1' &\dots & y_{N-1}'\end{bmatrix}^{\top}. $\\

\noindent
\textbf{Bi-level optimization problem.} Ideally, the malicious agent would like to implement the worst possible attack, solution of the following bi-level optimization problem:
\begin{equation}\label{eq:op_1}
\begin{aligned}
\max_{\boldsymbol{u}', \boldsymbol{y}'}\quad &\mathcal{A}(\set D_N,K(\boldsymbol{u}', \boldsymbol{y}'))& \\
\textrm{s.t.} \quad & K(\boldsymbol{u}', \boldsymbol{y}') \in \argmin_{K} \set L(\boldsymbol{u}', \boldsymbol{y}', K)\\
&\|\boldsymbol{u}'-\boldsymbol{u}\|_{q_u} \leq \delta_u,\quad \|\boldsymbol{y}'-\boldsymbol{y}\|_{q_y} \leq \delta_y,\\
& \|\boldsymbol{u}'-\boldsymbol{u}\|_0\leq \lfloor \rho_u N \rfloor, \quad \|\boldsymbol{y}'-\boldsymbol{y}\|_0\leq \lfloor \rho_y N \rfloor,\\
\end{aligned}
\end{equation}
where $\mathcal{A}(\set D_N, K(\boldsymbol{u}', \boldsymbol{y}'))$ denotes the objective criterion of the malicious agent and $q_u,q_y$ are convex norms. For example, a malicious agent may simply choose a \textit{max-min} type of attack, where $\set A=\set L$. Alternatively, she may choose to maximize the absolute value of the closed-loop eigenvalues or other performance metrics.

The {\it amplitude} of the attack is constrained by two elements: (i) the maximal fraction of actuators/sensors poisoned data, measured by $\boldsymbol{\rho}=(\rho_u, \rho_y)\in [0,1]^2$; (ii) the maximal amount of change $\boldsymbol{\delta}=(\delta_u,\delta_y)$ the malicious agent introduces in the dataset, where $\delta_u$ represents the amount of change in the residual of the control signal and respectively $\delta_y$ for the measurement signal. The constraint (i) is analogous to the assumption being used in classical machine Learning data-poisoning attacks \cite{biggio1,munoz2017towards,jagielski2018manipulating}, which is motivated by applications in which attackers can only reasonably control a small fraction of the transmitted data. 
The constraint (ii) is used to model the magnitude of the residuals, i.e., the amount of change in the dataset $\set D_N$ before and after the attack. 

  

Observe that in \ref{eq:op_1}, we have highlighted the dependency of the control law $K$ on the attack vector signals $\boldsymbol{a}_u,\boldsymbol{a}_y$. Furthermore, the malicious agent's objective $\set A$ does not, in general, directly depend on the attack vector signal since we want to evaluate the performance of the new controller $K(\boldsymbol{u}',\boldsymbol{y}')$ on a clean dataset.  Further note that the presence of the $0$-norm makes the problem non-convex. The latter can be cast to a mixed-integer program (MIP). Alternatively, one can relax the $0$-norm constraints using the $1$-norm.
\subsection{Computing the poisoning attack}
In general, computing the optimal attack signal is non trivial. When the inner problem is constrained, additional non-convexities are introduced in the bi-level optimization problem \cite{sinha2017review}. We may however look for local optima using a gradient ascent approach, as this is done in classical Data-Poisoning problems \cite{biggio1,munoz2017towards}. To that aim, we will focus our analysis in computing the gradients of $\set A$. 

\ifdefined\isextended
\begin{remark}\textit{There are also some major differences to classical data poisoning for classification problems: first,  there is no label for the data, which implies that we can not simply maximize the probability of classification error. Second, the problem involves two sets of data, the input data $\boldsymbol{u}$ and the output data $\boldsymbol{y}$. This makes the problem more complicated, since the optimal attack vector $\boldsymbol{a}_u$ (or $\boldsymbol{a}_y$) could depend on the attacked vector signals $\boldsymbol{y}', \boldsymbol{u}'$ in a complex way.}\\
\end{remark}\fi

\noindent
\textbf{Gradient computation.} In order to compute the gradient, one first needs to compute how the attack vector signals affect the control law. In case $K$ is parametrized by a vector $\theta \in \mathbb{R}^k$, since $\set A$ does not directly depend on $\boldsymbol{a}_u$ (observe that the adversary's objective function does not directly depend on the attack signal), one can derivate $\set A$ with respect to a vector signal (for example, $\boldsymbol{a}_u$) and obtain
\begin{equation}\label{eq:jacobian_J}
\nabla_{\boldsymbol{a}_u}\set A = (\nabla_{\boldsymbol{a}_u} \theta
)\nabla_{\theta} \set A  ,\end{equation}
where $\nabla_{\boldsymbol{a}_u} \theta$ is a $N\times n_k$ matrix and each row is the partial derivative of $\theta$ with respect to a specific input, i.e.,  $(\nabla_{\boldsymbol{a}_u} \theta)_{i,j} =  \partial_{a_{u,i}} \theta_j$. Furthermore, we can also easily compute second order terms (this would allow one to also  use  Newton methods). Denoting by $\partial_i$ the partial derivative with respect to $a_{u,i}$ (or $a_{y,i}$), we have
\begin{gather}\label{eq:hessian_J}
\begin{aligned}
(\nabla_{\boldsymbol{a}_u}^2\set A)_{i,j} \ &= \partial_i[\partial_j \set A]=\partial_i[(\partial_j \theta)^{\top} \nabla_\theta \set A],\\
&= (\partial_{i,j} \theta)^{\top} \nabla_\theta \set A + (\partial_j \theta)^{\top} \nabla_\theta^2 \set A (\partial_i \theta).
 \end{aligned}
 \end{gather}
We are still left with the problem of computing $\nabla_{\boldsymbol{a}_u}\theta$. This is one of the main issues, since computing how the control law is affected by the poisoning attack may be non-trivial. In case the inner problem $\set L$ is convex and sufficiently regular, then  it is possible to replace the inner problem with its stationary KKT conditions. This is also called \textit{Single-Level Reduction} \cite{sinha2017review}, and reduces the overall bi-level optimization problem to a single-level constrained problem by replacing the inner problem with $\nabla_{\theta}\set L(\boldsymbol{u}',\boldsymbol{y}', K_\theta)=0$. We then have the KKT conditions
\begin{equation}0=\frac{\textrm{d}}{\textrm{d}\boldsymbol{a}_u}\nabla_{\theta}\set L(\boldsymbol{u}',\boldsymbol{y}', K_\theta)=\nabla_{\boldsymbol{a}_u}\nabla_{\theta}\set L+\nabla_{\boldsymbol{a}_u}\theta \nabla_\theta^2 \set L.\end{equation}
In case  $\nabla_\theta^2 \set L$ is nonsingular, we can directly deduce that
\begin{equation}\label{eq:grad_theta_wrt_u}
\nabla_{\boldsymbol{a}_u}\theta=-(\nabla_{\boldsymbol{a}_u}\nabla_{\theta}\set L)(\nabla_\theta^2 \set L)^{-1}.\end{equation}
We can use the same reasoning to compute $\nabla_{\boldsymbol{a}_y}\theta$.
We have now all the ingredients to solve  problem \ref{eq:op_1} approximately. For example, one may use projected gradient ascent algorithms by iteratively updating the attack signals and project them back on the set of allowed perturbations $\set S_u \coloneqq \{\boldsymbol{x} \in \mathbb{R}^{N}: \|\boldsymbol{x}\|_{q_u} \leq \delta_u\}$ and $\set S_y \coloneqq \{\boldsymbol{x}\in \mathbb{R}^{N}: \|\boldsymbol{x}\|_{q_y} \leq \delta_y\}$.
\ifdefined\isextended \\ \fi
\begin{remark}\label{remark:norm_grad_a_theta} \textit{If $\|\nabla_{\boldsymbol{a}_u}\theta\|$ is sufficiently small, then it may be hard for the malicious agent to find an appropriate way to perturb the data using a gradient line search (since $\nabla_{\boldsymbol{a}_u}\set A$ would be small). One could enforce $\|\nabla_{\boldsymbol{a}_u}\theta\|$ to be small by making sure that  $\|\nabla_\theta^2 \set L\|$ is sufficiently big, which is related to the excitation persistence of the signals.}\end{remark}

Next we analyse how to attack a model-reference based data-driven control method, the Virtual Reference Feedback Tuning method. We leave the analysis of methods based on Willems et al. lemma to the technical report \cite{techreport}.

\section{Poisoning attacks on VRFT}
The Virtual Reference Feedback Tuning technique, together with the correlation approach, has been widely used in literature (and also on physical plants). It has been an active area of research for the last few decades \cite{campi2002virtual,formentin2012non,formentin2014comparison,hjalmarsson1998iterative,karimi2004iterative,karimi2007non,sala2005extensions,campestrini2011virtual,lequin2003iterative}. It shares similarities with the correlation approach, and therefore our analysis can be also applied to that method.  \ifdefined\isextended For simplicity, we will restrict our attention to the single-input/single-output case (a generalization of the VRFT method to MIMO systems can be found in \cite{formentin2012non}).\fi

We provide a generic analysis, and only introduce the malicious agent's objective at the end of the section. First, we explain how to compute $\nabla_{\boldsymbol{a}_u} \theta$ and $\nabla_{\boldsymbol{a}_y}\theta$. To that aim, we re-write the VRFT criterion in a convenient matrix form. Then, we provide a first set of analytical results quantifying the potential impact of attacks. Finally, we illustrate the analysis in the case where the malicious agent wishes to maximize the learner's loss, an attack referred to as {\it max-min} attack.

\subsection{Learner's loss}
We now introduce the attacked learner's cost criterion. As mentioned in Section 2, it is common practice to pre-filter the data using a filter $L(z)$. We assume for simplicity that the data has been already filtered, although it can be easily included in our analysis, and that Assumptions \ref{assumption1} and \ref{assumption2vrft_linear_control} hold. The cost function minimized by the learner is related to the $\ell_2$ norm of the control signal, specifically it is
\begin{align}\set L(\boldsymbol{u}',\boldsymbol{y}', K_\theta)) &= \frac{1}{N} \sum_{t=0}^{N-1} \|u_t' - K_\theta(z)\tilde e_t\|_2^2,\\
&= \frac{1}{N} \sum_{t=0}^{N-1} \|u_t+a_{u,t} - \beta(z)^{\top}\tilde e_t\theta\|_2^2,
\end{align}
where $\tilde e_t = r_t-y_t'=(M_r^{-1}(z)-1)(y_t+a_{y,t})$.
Now, rewriting the VRFT criterion in matrix form, we will be able to compute $\nabla_{\boldsymbol{a}_u}\theta$ and $\nabla_{\boldsymbol{a}_y}\theta$ using \ref{eq:grad_theta_wrt_u}. Let  $\phi_{t,i} = \beta_i(z) \tilde{e}_t$, $i=1,\dots, n_k$ and $\boldsymbol{\phi}_i =\begin{bmatrix}
\phi_{0,i},\dots,\phi_{N-1,i}
\end{bmatrix}^{\top}\in \mathbb{R}^N$. Then, it is possible to rewrite the VRFT criterion as
\begin{equation}\set L(\boldsymbol{u}',\boldsymbol{y}', K_\theta) = \frac{1}{N} \left \|\boldsymbol{u}'-\Phi(\boldsymbol{y}')\theta \right\|_2^2\end{equation}
where $\Phi(\boldsymbol{y}')=\begin{bmatrix}
\boldsymbol{\phi}_1,\dots, \boldsymbol{\phi}_{n_k}
\end{bmatrix} \in \mathbb{R}^{N\times n_k}$ is a matrix that containts the output response of the control law, and depends on $\boldsymbol{y}'$ since the error signal depends on $\boldsymbol{y}'$.
\subsection{Computing the gradients}
To be able to compute the gradients of the attack, we need one additional ingredient.  We denote the input-output response of a generic transfer function $G(z)$  over $[N]$ as $\boldsymbol{y}_{[N]}= \set T_{G,N}\boldsymbol{u}_{[N]} + \set O_{G,N} x_0$, where 
 \[\set T_{G,N} =\begin{bmatrix}
 D & 0 & 0 &\dots  & 0\\
 CB & D &0 & \dots & 0\\
 CAB & CB & D & \dots & 0\\
 \vdots & \vdots & \vdots &\ddots & \vdots\\
 CA^{N-1}B & CA^{N-2}B & CA^{N-3}B  &\dots & D
 \end{bmatrix}, \]
 represents the Toeplitz matrix of order $N+1$ of the system, and $\set O_{G,N} = \begin{bmatrix}
 C & CA & \dots & CA^{N}
 \end{bmatrix}^{\top}$ is the observability matrix of order $N+1$.  We can then derive the following lemma
\begin{lemma}
	Consider the VRFT criterion and assume without loss of generality that the relationship $r_t=M_r^{-1}(z)y_t$, holds with zero initial and final conditions. Then we have:
	\begin{align}
		\nabla_{\boldsymbol{a}_u}\theta &=  -\Phi(\boldsymbol{y}')\left(\Phi^{\top}(\boldsymbol{y}') \Phi(\boldsymbol{y}') \right)^{-1},\\
		\nabla_{\boldsymbol{a}_y}\theta &= -S\left(\Phi^{\top}(\boldsymbol{y}') \Phi(\boldsymbol{y}') \right)^{-1} ,
	\end{align}
with \[(S)_{i,*} = \boldsymbol{u}^{'\top}TD_{N,n_k}(e_i)-\theta^{\top}C_i\] where $C_i =\Phi^{\top}(\boldsymbol{y}')TD_{N,n_k}(e_i) + D_{N,n_k}^{\top}(e_i)^{\top}T^{\top} \Phi(\boldsymbol{y}')$. $T\in\mathbb{R}^{N\times Nn_k}$ is the overall Toeplitz matrix of the control signal (from the output signal)
\begin{equation}
T = \begin{bmatrix}
	\set T_{\beta_1}(\set T_{M_r^{-1}} - I_N) & \dots & \set T_{\beta_{n_k}}(\set T_{M_r^{-1}} - I_N)
\end{bmatrix}
\end{equation}
where $\set T_{M_r^{-1}}$ denotes the Toeplitz matrix of the system $r_t=M_r^{-1}(z)y_t$ and $D_{p,q}: \mathbb{R}^p \to \mathbb{R}^{pq\times q}$ is a generalized diagonalization operator 
\begin{equation}D_{p,q}(\boldsymbol{x}) = \begin{bmatrix}
		\boldsymbol{x} & \boldsymbol{0} & \dots & \boldsymbol{0}\\
		\boldsymbol{0} & \boldsymbol{x} & \dots & \boldsymbol{0}\\
		\vdots &\vdots &\ddots &\vdots\\
		\boldsymbol{0} & \boldsymbol{0} & \dots & \boldsymbol{x}
	\end{bmatrix},\quad \boldsymbol{0}=\begin{bmatrix} 0\\\vdots \\0\end{bmatrix}\in \mathbb{R}^p.\end{equation}
\\
\end{lemma}
\begin{proof} To find $\nabla_{\boldsymbol{a}_u}\theta$ and $\nabla_{\boldsymbol{a}_y}\theta$, we use \ref{eq:grad_theta_wrt_u}. We can easily find that $\nabla_\theta \set L= -\frac{2}{N}\Phi^{\top}(\boldsymbol{y}') (\boldsymbol{u}'-\Phi(\boldsymbol{y}') \theta)$ and $\nabla^2_\theta \set L=\frac{2}{N} \Phi^{\top}(\boldsymbol{y}')  \Phi(\boldsymbol{y}')$, with minimum of $\set L$ given by $\hat{\theta}(\boldsymbol{u}', \boldsymbol{y}')$:
\begin{equation}\label{eq:solution_theta}\hat{\theta}(\boldsymbol{u}', \boldsymbol{y}') = \left(\Phi^{\top}(\boldsymbol{y}') \Phi(\boldsymbol{y}') \right)^{-1}\Phi^{\top}(\boldsymbol{y}') \boldsymbol{u}'.\end{equation}
\textbf{Computation of $\nabla_{\boldsymbol{a}_u}\theta$. }
Since $\nabla_{\boldsymbol{a}_u}\nabla_\theta \set L= -\frac{2}{N} \Phi(\boldsymbol{y}')$, then, using $\nabla^2_\theta \set L=\frac{2}{N} \Phi^{\top}(\boldsymbol{y}')  \Phi(\boldsymbol{y}')$ we have
\begin{equation}\nabla_{\boldsymbol{a}_u} \theta=-\Phi(\boldsymbol{y}')\left(\Phi^\top(\boldsymbol{y}') \Phi(\boldsymbol{y}') \right)^{-1},\end{equation}
which is remarkably independent of $\boldsymbol{u}'$.  Therefore for a persistently exciting signal $\boldsymbol{y}$, we may expect $\nabla_{\boldsymbol{a}_u} \theta$ to be small in norm.
\\\\
\textbf{Computation of $\nabla_{\boldsymbol{a}_y}\theta$. } Computing $\nabla_{\boldsymbol{a}_y}\theta$ is more involved. Actually, we only need to compute the quantity $\nabla_{\boldsymbol{a}_y}\nabla_\theta \set L$. To do so, we need to understand how $\boldsymbol{\phi}_i$ is affected by the signal $\boldsymbol{y}'$.

The output $\boldsymbol{\phi}_i$ of the i-th controller can also be rewritten as the output response over $[N-1]$ of $\beta_i(z)$ given the input $\tilde{\boldsymbol{e}}$. As mentioned in Section 2, we can write 
$\boldsymbol{\phi}_i = \set T_{\beta_i,N}\tilde{\boldsymbol{e}} + \set O_{ \beta_i,N}\beta_{i,0}$ where $\tilde{\boldsymbol{e}} = \boldsymbol{r}-\boldsymbol{y}'$. The initial conditions will not affect the analysis, and they can be assumed to be 0. We will also avoid the subscript $N$ for simplicity. We can write $\Phi$ in the following manner
\begin{equation}\Phi = \begin{bmatrix}\set T_{\beta_1} (\boldsymbol{r}-\boldsymbol{y}')& \dots & \set T_{\beta_{n_k}} (\boldsymbol{r}-\boldsymbol{y}')\end{bmatrix}   \end{equation}
with
$\theta= \sum_{i=1}^{n_k} \set T_{\beta_i}\theta_i \tilde{\boldsymbol{e}}=\sum_{i=1}^{n_k} \set T_{\beta_i}\theta_i (\boldsymbol{r}-\boldsymbol{y}').$
We also need to make the relationship between $\boldsymbol{r}$ and $\boldsymbol{y}$ explicit. Since $r_t = (M_r^{-1}(z)-1)y_t$ we can study $M_r^{-1}(z)$. Suppose the state space realization of $M_r(z)$ is given by $(A,B,C,D)$: if $D$ is non singular or $F=CA^{-1}B$ is nonsingular, then one can use Theorem 1 in \cite{kavranoglu1993new} to invert the state space formulation and obtain a linear relationship of the type $\boldsymbol{r}=\set \set T_{M_r^{-1}} \boldsymbol{y}$, where $\set T_{M_r^{-1}}$ is the Toeplitz matrix of order $N$ of the inverted system (where we assumed zero initial/final conditions).
\ifdefined\isextended
For example, if $F$ is nonsingular then we can write the state space representation of $M_r^{-1}(z)$ as follows
\begin{align}
x_{t-1}=\bar A x_t + \bar B y_t,\quad r_t =  \bar C x_t + \bar D y_t
\end{align}
where
\begin{align}
\bar A &= A^{-1}(I-BF^{-1}CA^{-1}),\quad \bar C =-F^{-1}CA^{-1},\\
\bar B &= -A^{-1}BF^{-1},\quad \bar D = -F^{-1},
\end{align}
Assuming zero final conditions this allows one to write $\boldsymbol{r}=\set T_{M_r^{-1}} \boldsymbol{y}$, where $\set T_{M_r^{-1}}$ is
\begin{equation}
\set T_{M_r^{-1}} = \begin{bmatrix}
\bar D &\dots & \bar C\bar A^{N-3}\bar B & \bar C\bar A^{N-2}\bar B\\
\vdots  &\ddots &\vdots &\vdots \\
0 & \dots & \bar D & \bar C \bar B\\
0 & \dots & 0 & \bar D
\end{bmatrix}.
\end{equation}
\fi
In case $CA^{-1}B$ is singular, one can still try to express the reference signal in regressor form, and write $r_t = \rho_t^{\top} \boldsymbol{y}$ for some $\rho_t \in \mathbb{R}^N$, so that $\set T_{M_r^{-1}} = \begin{bmatrix}
\rho_0^{\top} & \dots &\rho_{N-1}^{\top}
\end{bmatrix}^{\top}.$ This allows us to express the virtual error vector as $\boldsymbol{e} = (\set T_{M_r^{-1}} - I ) \boldsymbol{y}$. If we now define the overall Toeplitz matrix from the output to the control signal $ T = \begin{bmatrix}
	\set T_{\beta_1}(\set T_{M_r^{-1}} - I) & \dots & \set T_{\beta_{n_k}}(\set T_{M_r^{-1}} - I)
\end{bmatrix}$, then we can write $\Phi(\boldsymbol y) = T D_{N,n_k}(\boldsymbol{y})$.
\ifdefined\isextended
 where $D$ is a generalized diagonalization operator $D_{p,q}: \mathbb{R}^p \to \mathbb{R}^{pq\times q}$
\begin{equation}D_{p,q}(\boldsymbol{x}) = \begin{bmatrix}
\boldsymbol{x} & \boldsymbol{0} & \dots & \boldsymbol{0}\\
\boldsymbol{0} & \boldsymbol{x} & \dots & \boldsymbol{0}\\
\vdots &\vdots &\ddots &\vdots\\
\boldsymbol{0} & \boldsymbol{0} & \dots & \boldsymbol{x}
\end{bmatrix},\quad \boldsymbol{0}=\begin{bmatrix} 0\\\vdots \\0\end{bmatrix}\in \mathbb{R}^p.\end{equation}
\else
\fi
We can find the derivative with respect to $y_j$: $\frac{\partial }{\partial y_j} \Phi(\boldsymbol y)=  T D(e_j)$.
  Similarly, we have $\frac{\partial }{\partial y_j} \Phi\theta= T D(e_j)\theta$. Let $C_j=\Phi^{\top}(\boldsymbol{y}')TD(e_j) + D^{\top}(e_j)^{\top}T^{\top} \Phi(\boldsymbol{y}')$, it follows that:
\begin{equation}
\begin{aligned}\Big(\nabla_{\boldsymbol{a}_y}\nabla_\theta \set L\Big)_{j,*}&= -\frac{2}{N} \Big[(\boldsymbol{u}'-\Phi(\boldsymbol{y}') \theta)^{\top}TD(e_j) \\
	&\quad\quad -(TD(e_j)\theta)^{\top}\Phi(\boldsymbol{y}')\Big],\\
&= -\frac{2}{N} (\boldsymbol{u}^{'\top}TD(e_j)-\theta^{\top}C_j).
\end{aligned}.
\end{equation}
\end{proof}
\ifdefined\isextended
\begin{remark}\textit{
In the previous lemma we made use of the assumption that we could write $\boldsymbol{r}=\set \set T_{M_r^{-1}} \boldsymbol{y}$, where $T_{M_r^{-1}} $ is a square matrix.  This assumption holds in case the initial/final conditions are $0$ for $r_t=M_r^{-1}(z)$. If that is not the case, then we need to augment the matrix $T_{M_r^{-1}}$ with additional columns, as many as needed in order to take into account the extra conditions. Obviously, we assume that the user also collected this extra data from experiments.}\ifdefined\isextended \\ \fi
\end{remark}
\fi

\subsection{Impact of poisoning attacks} 

Using the above analysis, we can quantify the potential impact of poisoning attacks. More precisely, we can upper bound the difference between $\theta = \hat{\theta}(\boldsymbol{u}, \boldsymbol{y})$ and $\theta' = \hat{\theta}(\boldsymbol{u}', \boldsymbol{y}')$ before and after the attack. 

\begin{lemma}\label{lemma:diff_bound}
Let $(\boldsymbol{a}_u, \boldsymbol{a}_y)$ be a generic data-poisoning attack, with constraints $\|\boldsymbol{a}_u\|\leq \delta_u$ and $\|\boldsymbol{a}_y\| \leq \delta_y$. Then, we have:
\begin{equation}\left \|\theta-\theta' \right\|_2 \leq \gamma \sqrt{n_k}(\|\boldsymbol{y}\|_2\delta_u + \|\boldsymbol{u}\|_2\delta_y),\end{equation}
where
$\gamma\coloneqq \sigmamax(T)/\sigmamin(\Phi(\boldsymbol{y})^{\top}\Phi(\boldsymbol{y})).$
\end{lemma}
\begin{proof}
For ease of exposition, let $\Phi=\Phi(\boldsymbol{y}), \tilde\Phi=\Phi(\boldsymbol{a}_y)$ and $\Phi'=\Phi(\boldsymbol{y}')=\Phi+\tilde\Phi$. Define also $P=\Phi^{\top}(\boldsymbol{y})\Phi(\boldsymbol{y})$ and $ P' =  \Phi(\boldsymbol{y}')^{\top}  \Phi(\boldsymbol{y}')$. Using the identity $(A+B)^{-1} = A^{-1}-A^{-1}B(A+B)^{-1}$ and by noting that $P'=P + \Lambda(\boldsymbol{y}')$ where $\Lambda = \Phi(\boldsymbol{a}_y)^{\top}  \Phi(\boldsymbol{a}_y)+ 2\Phi(\boldsymbol{a}_y)^{\top}  \Phi(\boldsymbol{y})$, we obtain
	$(P')^{-1} = P^{-1}-P^{-1}\Lambda (P')^{-1}$. This identity allows us to work out the difference between $\theta$ and $\theta'$:
	\begin{align*}
		\theta-\theta' &= P^{-1}\Phi^{\top} \boldsymbol{u} - (P')^{-1}\Phi^{'\top} \boldsymbol{u}'\\
				       &\stackrel{(a)}{=} P^{-1}[\Phi^{\top} \boldsymbol{u} - (I-\Lambda (P')^{-1})(\Phi+\tilde\Phi)^{\top} (\boldsymbol{u}+\boldsymbol{a}_u)]\\
				       &= P^{-1}[-(I-\Lambda (P+\Lambda)^{-1})(\tilde\Phi^{\top}\boldsymbol{u}+(\Phi')^{\top}\boldsymbol{a}_u)]\\
				       &\stackrel{(b)}{=} P^{-1}[-P(P+\Lambda)^{-1}(\tilde\Phi^{\top}\boldsymbol{u}+(\Phi')^{\top}\boldsymbol{a}_u)]\\
				       &=-(P+\Lambda)^{-1}(\tilde\Phi^{\top}\boldsymbol{u}+(\Phi')^{\top}\boldsymbol{a}_u)
	\end{align*}
	where in (a) we used the inverse matrix identity and in (b) we factored out $(P+\Lambda)^{-1}$. 	Now, observe that for a generic vector $\boldsymbol{x}$ we could write $\Phi(\boldsymbol{x}) = TD(\boldsymbol{x})$, from which follows that $\|D(\boldsymbol{x})\|_2^2 = n_k\|\boldsymbol{x}\|_2^2$ and as a consequence $\|\Phi(\boldsymbol{x})\|_2 \leq \sqrt{n_k} \|T\|_2\|\boldsymbol{x}\|_2$. Using the previous result on  $\|\theta-\theta'\|_2$ leads to the following upper bound
	\[\|\theta-\theta'\|_2 \leq \sqrt{n_k}\|(P+\Lambda)^{-1}\|_2 \|T\|_2 (\|\boldsymbol{u}\|_2\delta_y + \|\boldsymbol{y}\|_2\delta_u). \]
	 We  conclude by observing that combining $\sigmamax((P+\Lambda)^{-1}) = 1/\sigmamin(P+\Lambda)$ and $\sigmamin(P+\Lambda) \geq \sigmamin(P)+\sigmamin(\Lambda)$ yields  $\|(P+\Lambda)^{-1}\|_2 \leq \frac{1}{\sigmamin(P)+\sigmamin(\Lambda)}\leq \frac{1}{\sigmamin(P)}.$
\end{proof}

What Lemma \oldref{lemma:diff_bound} tells us is that for a small value of $\gamma$, the difference in the two parameters will be small. This is in line with Remark \oldref{remark:norm_grad_a_theta}: if the signals are sufficiently exciting, then the minimum singular value of $P$ will be large, from which follows that $\gamma$ will be small. As a result,  the malicious agent, to be able to affect the parameter $\theta$, will have to use signals of larger magnitude, which are easier to detect. Moreover, the difference can also be minimized by reducing $n_k$ (the number of parameters). The next result, whose proof can be found in our technical report \cite{techreport}, confirms the importance of using persistently exciting signals.
\begin{proposition}\label{corollary:targeted_theta}
If for a fixed $\boldsymbol{y}'$ the condition $\rank(\Phi(\boldsymbol{y}')) \geq n_k$ holds, then for $\delta_u$ sufficiently large the malicious agent can choose any parameter $\theta \in \mathbb{R}^{n_k}$ by selecting a proper signal $\boldsymbol{a}_u$.
\end{proposition}
\ifdefined\isextended
\begin{proof} 
The proof follows from the fact that if $\rank(\Phi(\boldsymbol{y}')) \geq n_k$ holds then $\Phi(\boldsymbol{y}')^{\top}\Phi(\boldsymbol{y}')$ is full rank, and $\rank((\Phi(\boldsymbol{y}')^{\top}\Phi(\boldsymbol{y}'))^{-1}\Phi(\boldsymbol{y}')^{\top})=\rank(\Phi(\boldsymbol{y}')^{\top}) \geq n_k$.  Then, for any $x \in \mathbb{R}^{n_k}$ there exists $\boldsymbol{a}_u \in \mathbb{R}^{N}$ such that $x =(\Phi(\boldsymbol{y}')^{\top}\Phi(\boldsymbol{y}'))^{-1}\Phi(\boldsymbol{y}')^{\top}(\boldsymbol{u}+\boldsymbol{a}_u)$, i.e., $\boldsymbol{a}_u = \Phi(\boldsymbol{y}')^\dagger (\Phi(\boldsymbol{y}')^{\top}\Phi(\boldsymbol{y}'))x-\boldsymbol{u}$. Therefore, for any $x \in \mathbb{R}^{n_k}$ there exists $\varepsilon>0$ such that  for $\delta_u\geq \varepsilon$ the malicious agent can craft an attack that will make the parameter vector $\hat{\theta}(\boldsymbol{u}',\boldsymbol{y}')$ equal to $x$.
\end{proof}\fi

The previous proposition indicates that using a well-crafted attack, the malicious agent can dictate the closed-loop behavior of the system. For example, she may want $\hat \theta$ to converge to some chosen parameter $\theta_a$. The malicious agent in this case may not even need to change $\boldsymbol{y}$, and can just  choose an appropriate vector $\boldsymbol{a}_u$. Since $\hat{\theta}(\boldsymbol{u}', \boldsymbol{y}') = \left(\Phi^{\top}(\boldsymbol{y}') \Phi(\boldsymbol{y}') \right)^{-1}\Phi^{\top}(\boldsymbol{y}') \boldsymbol{u}'$, if the malicious agent wants $\hat{\theta}$ to converge to $\theta_a$, then the lower is the norm of $\left(\Phi^{\top}(\boldsymbol{y}') \Phi(\boldsymbol{y}') \right)^{-1}$ the higher needs to be $\delta_u$. It is well-known that for persistently exciting signals, this norm will be small, which in turn implies that $\delta_u$ needs to be bigger. 
\ifdefined\isextended
\begin{remark}\textit{If one wishes to include the effect of the filter $L(z)$ in the analysis, then we just need to include the effect of $\set T _L$ (the Toeplitz matrix of order $N$ of $L(z)$) in $T=\begin{bmatrix}
	\set T_{\beta_1}\set T_{L}(\set T_{M_r^{-1}}-I)&\dots &\set T_{\beta_{n_k}}\set T_{L}(\set T_{M_r^{-1}}-I)
	\end{bmatrix}$.
} 
\end{remark}
\fi


\subsection{Max-min attack}

In a max-min attack, the malicious agent aims at maximizing the learner's loss. Formally, we just have: $\set A= \set L$. By choosing this cost function, the malicious agent is implicitly maximizing the residual error $\|M_r(z)-[(I-M_r)GK_\theta](z)\|_2$ (as $N\to \infty$). This cost function may seem attractive, but it does not allow the malicious agent to explicitly control the dynamics of the closed-loop system. For example, the resulting closed-loop system may still be stable after the poisoning attack. Under the assumption that the malicious agent can affect the entire dataset, the optimization problem (to devise an optimal attack) is as follows:
\begin{equation}\label{eq:op_maxmin}
\begin{aligned}
\max_{\boldsymbol{u}' \boldsymbol{y}'}\quad &\set A(\boldsymbol{u},\boldsymbol{y}, \hat \theta(\boldsymbol{u}', \boldsymbol{y}'))\coloneqq \frac{1}{N}\left\|\boldsymbol{u}-\Phi(\boldsymbol{y})\hat \theta(\boldsymbol{u}', \boldsymbol{y}')\right\|_2^2& \\
\textrm{s.t.} \quad & \hat{\theta}(\boldsymbol{u}', \boldsymbol{y}') = \left(\Phi^{\top}(\boldsymbol{y}') \Phi(\boldsymbol{y}') \right)^{-1}\Phi^{\top}(\boldsymbol{y}') \boldsymbol{u}'\\
&\|\boldsymbol{u}'-\boldsymbol{u}\|_{q_u} \leq \delta_u,\quad \|\boldsymbol{y}'-\boldsymbol{y}\|_{q_y} \leq \delta_y.
\end{aligned}
\end{equation}

\noindent
\textbf{Convexity.} The above max-min optimization problem enjoys the following nice property. The objective function is convex in $\boldsymbol{u}'$: this is due to \ref{eq:solution_theta} providing the expression of $\hat{\theta}$. Therefore, for a fixed vector $\boldsymbol{y}'$, the maximum over $\boldsymbol{u}'$ is attained on some extremal point of the feasible set, specifically for $\|\boldsymbol{u'}-\boldsymbol{u}\|_{q_u}=\delta_u$. To find the optimal attack vector on the input signal, one can use disciplined convex-concave programming (DCCP) \cite{shen2016disciplined}. Convexity with respect to $\boldsymbol{y}'$ does not hold, which can be easily verified for the simple case $n_k=1$ and $N=2$. Note however that we can easily compute $\nabla_{\boldsymbol{a}_u} \set A$ and $\nabla_{\boldsymbol{a}_y} \set A$ using the expressions found previously (and also second-order terms).

\begin{algorithm}[h]
  \DontPrintSemicolon
  \KwIn{Dataset $\set D_N=(\boldsymbol{u},\boldsymbol{y})$; objective function $\set A$; parameters $\boldsymbol{\delta}$, $\eta$}
  \KwOut{Attack vectors $\boldsymbol{a}_u, \boldsymbol{a}_y$}
  $i \gets 0, (\boldsymbol{a}_u^{(i)},\boldsymbol{a}_y^{(i)})\gets (\boldsymbol{0},\boldsymbol{0})$ \Comment*[r]{Initialize algorithm}
  $\hat\theta^{(i)}\gets \hat{\theta}( \boldsymbol{u}+\boldsymbol{a}_u^{(i)}, \boldsymbol{y}+\boldsymbol{a}_y^{(i)})$ where $\hat\theta$ is given in \ref{eq:op_maxmin}\;
  $J^{(i)} \gets \set A(\boldsymbol{u},\boldsymbol{y}, \hat\theta^{(i)})$\;
  \Do{$|J^{(i+1)}-J^{(i)}|>\eta$}{
		 $\boldsymbol{a}_u^{(i+1)} \gets $ solve \ref{eq:op_maxmin} in $\boldsymbol{a}_u$ using DCCP \cite{shen2016disciplined}\;
		 $\boldsymbol{a}_y^{(i+1)} \gets \textsc{PGA}(\delta_y, \gamma_i,\hat{\theta}( \boldsymbol{u}+\boldsymbol{a}_u^{(i+1)}, \boldsymbol{y}+\boldsymbol{a}_y^{(i)}))$\; 
	
        $\hat\theta^{(i+1)}\gets \hat{\theta}( \boldsymbol{u}+\boldsymbol{a}_u^{(i+1)}, \boldsymbol{y}+\boldsymbol{a}_y^{(i+1)})$\;
        $J^{(i+1)}\gets \set A(\boldsymbol{u},\boldsymbol{y}, \hat\theta^{(i+1)})$\;
        $i \gets i+1$\;
      }
  \Return $(\boldsymbol{a}_u^{(i)}, \boldsymbol{a}_y^{(i)})$
  \caption{Max-min attack algorithm}
  \label{algo1}
\end{algorithm}

\noindent
\textbf{Upper bound on $\set A$. } Using Lemma \oldref{lemma:diff_bound}, we can derive an upper bound on $\set A$, which can help to quantify the maximal impact of max-min attacks.
\begin{corollary}\label{corollary:minmax_bound}
The objective function of the optimization problem \ref{eq:op_maxmin} satisfies:
\begin{align*}
\set A(\boldsymbol{u},\boldsymbol{y}, \hat \theta(\boldsymbol{u}', &\boldsymbol{y}')) \leq \Bigg(\sqrt{ \set L(\boldsymbol{u},\boldsymbol{y}, K_{\hat \theta(\boldsymbol{u}, \boldsymbol{y})})}\\ &+   \frac{n_k\|\boldsymbol{y}\|_2 \sigmamax^2(T)}{\sigmamin(P)} (\|\boldsymbol{y}\|_2\delta_u + \|\boldsymbol{u}\|_2\delta_y)\Bigg )^2.
\end{align*} 
\end{corollary}
\begin{proof}
\ifdefined\isextended
A simple use of triangular inequality allows us to deduce the following
\begin{align*}
\set A(\boldsymbol{u},\boldsymbol{y}, \hat \theta(\boldsymbol{u}', \boldsymbol{y}'))  &= \left\|\boldsymbol{u}-\Phi(\boldsymbol{y})\hat \theta(\boldsymbol{u}', \boldsymbol{y}')\right\|_2^2,\\
&\leq \left(\left\|\boldsymbol{u}-\Phi(\boldsymbol{y})\hat \theta(\boldsymbol{u}, \boldsymbol{y})\right\|_2 +  \left\|\Phi(\boldsymbol{y})\Delta\theta\right\|_2\right)^2
\end{align*}
where $\Delta\theta = \hat \theta(\boldsymbol{u}', \boldsymbol{y}')-\hat \theta(\boldsymbol{u}, \boldsymbol{y})$. By lemma \ref{lemma:diff_bound} and $\|\phi(\boldsymbol{y})\|_2 \leq \sqrt{n_k}\sigmamax(T) \|\boldsymbol{y}\|_2$ we get
\[\set A\leq \left(\sqrt{ \set L} + n_k\|\boldsymbol{y}\|_2  \frac{\sigmamax^2(T)}{\sigmamin(P)} (\|\boldsymbol{y}\|_2\delta_u + \|\boldsymbol{u}\|_2\delta_y)\right )^2 .\]
\else
The result directly follows from combining the triangular inequality and Lemma  \oldref{lemma:diff_bound}.
\fi
\end{proof}

\noindent
\textbf{Algorithm. }
In Algorithm \ref{algo1}, we propose an alternating gradient ascent algorithm to compute the attack. Since we can compute the optimal attack on the input signal using DCCP \cite{shen2016disciplined}, the idea is to first compute $\boldsymbol{a}_u$ (keeping $\boldsymbol{a}_y$ fixed) and then compute $\boldsymbol{a}_y$ according to $\boldsymbol{a}_u$, using a projected gradient ascent (PGA) search, where we iteratively update $\boldsymbol{a}_y$ along $\nabla_{\boldsymbol{a}_y}\set A$ and project onto $S_y$. At the end of each iteration $k$, we compute the difference $\Delta_i J\coloneqq  J^{(i+1)} - J^{(i)}$, where $J^{(i)}=\set A(\set D_N, \hat{\theta}(\boldsymbol{u}+\boldsymbol{a}_u^{(i)},\boldsymbol{y}+\boldsymbol{a}_y^{(i)} ))$, and we stop if $|\Delta_i J| <\eta$, with $\eta>0$ being a user-chosen parameter. In practice, the difficulty lies in computing the PGA step (it requires some fine tuning to adjust the step size).

%% file: 4.simulations.tex
\section{Numerical simulations}
\textbf{Plant dynamics and VRFT method.} We consider the numerical example proposed in \cite{campi2002virtual}  of a 
flexible transmission system,  which was originally proposed in \cite{landau1995flexible} as a benchmark for digital control design. The continuous plant is discretized with sampling time $T_s=0.05s$, and the dynamics are given by
$G(z)=B(z)/A(z)$
with $A(z)=z^4-1.41833z^{3}+1.58939z^{2}-1.31608z+0.88642$ and $B(z)=0.28261z+0.50666$. The reference model is $M_r(z)= (1-\alpha)^2/z(z-\alpha)^2$ with $\alpha=e^{-T_s\bar \omega},\bar \omega=10. $
For such system, the class of controllers considered is  of the form $\beta_i(z)=z^{2-i}/(z-1), i=1,\dots,6$. We will consider two types of input signals: (A)  a step function $u_t$ that is equal to $1$ for $t\in [5,15]$  and $0$ otherwise; (B) a persistently exciting input signal $u_t \sim \mathcal{N}(0,1)$.  In both cases, $N=512$ data points are collected from the plant and the data has been low pass filtered using $L(z)=(1-M_r(z))M_r(z)$ as in \cite{campi2002virtual}.

\medskip
\noindent
\textbf{Attack setting.}
We analyze the attack objective \ref{eq:op_maxmin} for different values of $\delta_u$ and $\delta_y$ using the euclidean norm. We used $\delta_u=\varepsilon_u\|\boldsymbol{u}\|_2$ with $\varepsilon_u \in \{0, 0.1, 0.2\}$, and $\delta_y = \varepsilon_y \|\boldsymbol{y}\|_2$ with $\varepsilon_y \in [0.01, 0.057]$. These values are chosen based on empirical evidence (we  observed that a slight increase of $\varepsilon_y$ brings a larger change in the closed-loop performance of the system).
\renewcommand{\thefigure}{2}
\begin{figure}[b]
	\centering
	\includegraphics[width=0.5\textwidth]{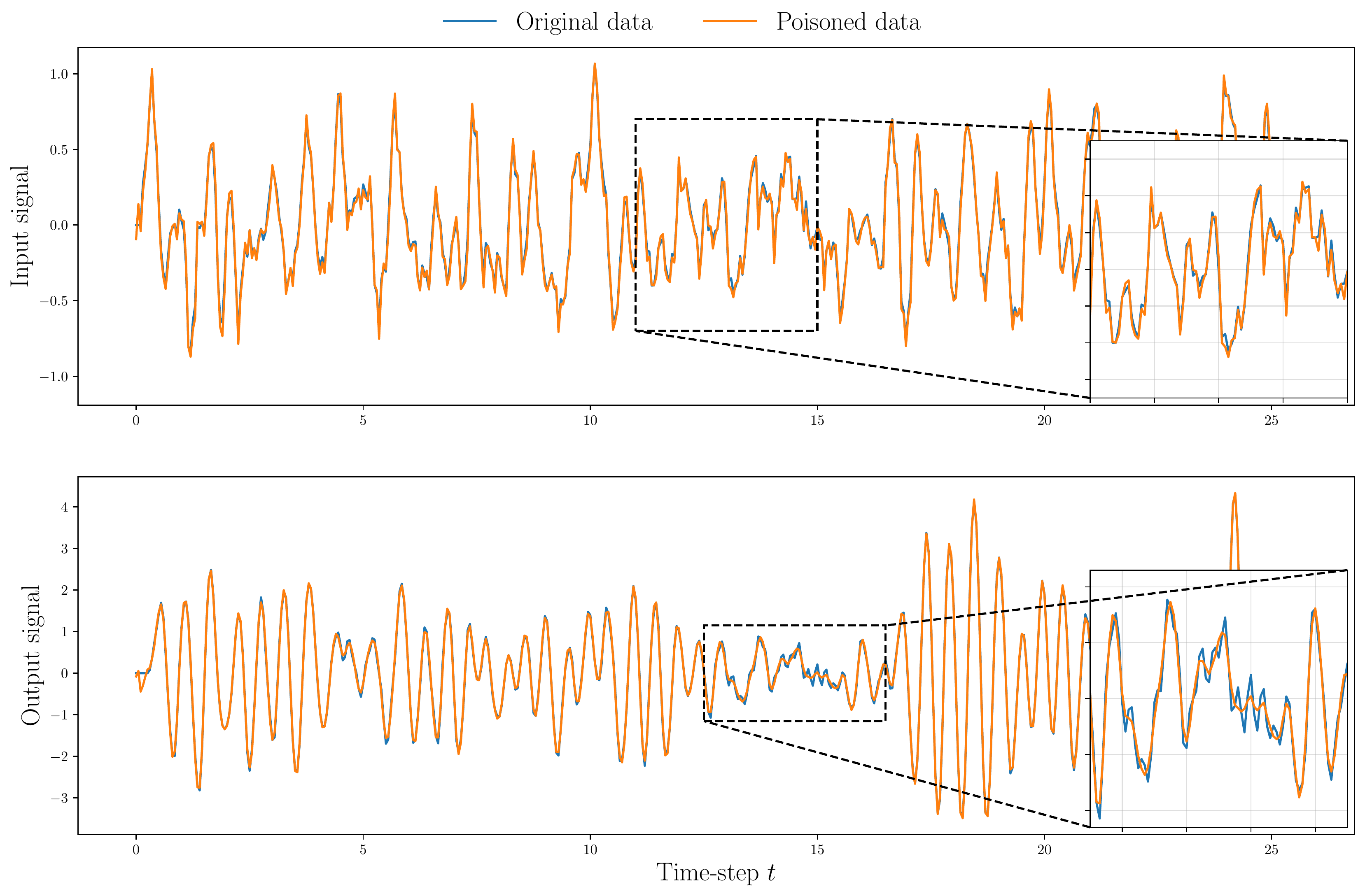}
	\caption{Input/output dataset $\set D_N$ for scenario (B) with $\varepsilon_u=0.1$ and $\varepsilon_y=0.057$.}
	\label{fig:input_output_data_whitenoise}
\end{figure}
\renewcommand{\thefigure}{3}
\begin{figure*}[]
	\centering
	\begin{subfigure}{0.49\linewidth}
		\includegraphics[width=\textwidth]{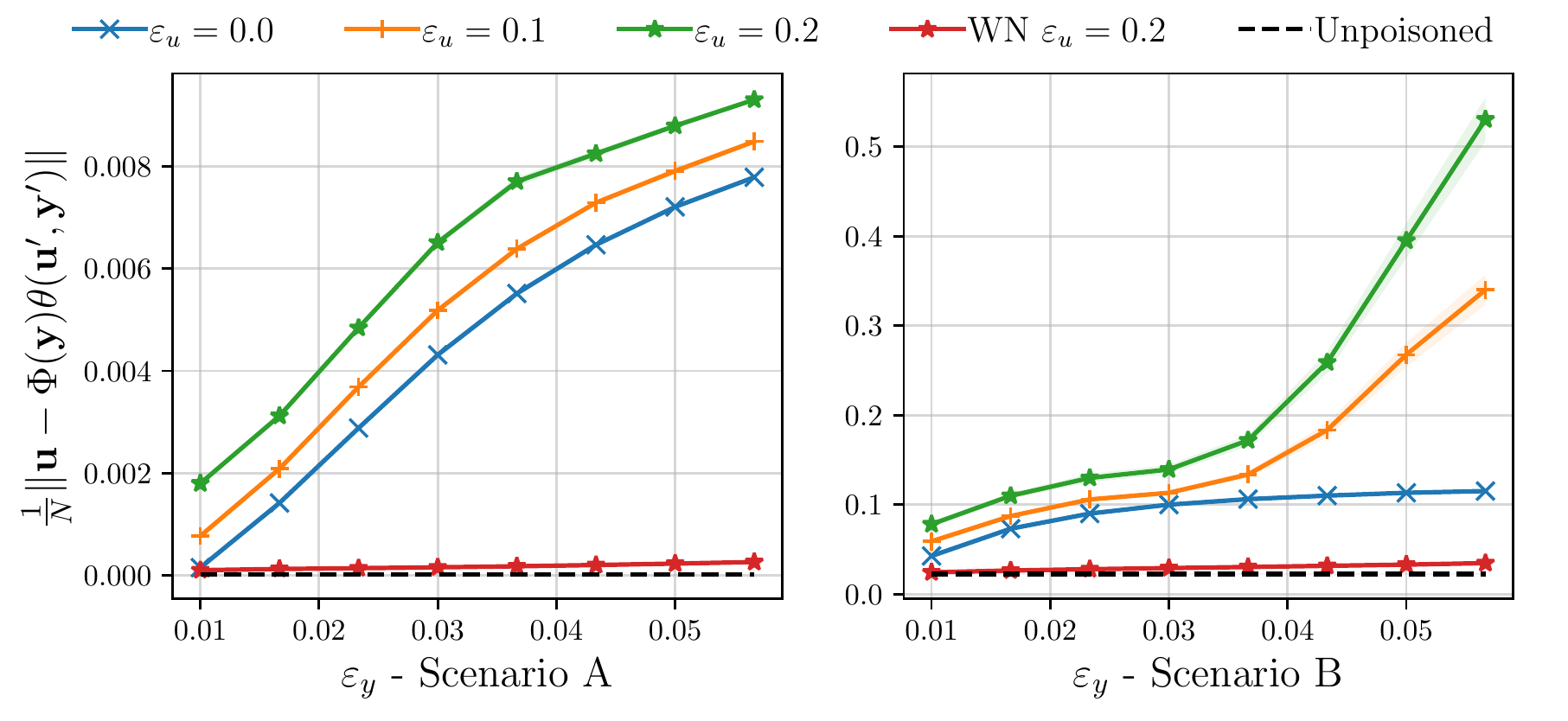}
		\caption{Fig.3: Loss of the learner $\set L$ for different values of $\varepsilon_u,\varepsilon_y$.}
		\label{fig:loss}
	\end{subfigure}
	\begin{subfigure}{0.49\linewidth}
		\includegraphics[width=\textwidth]{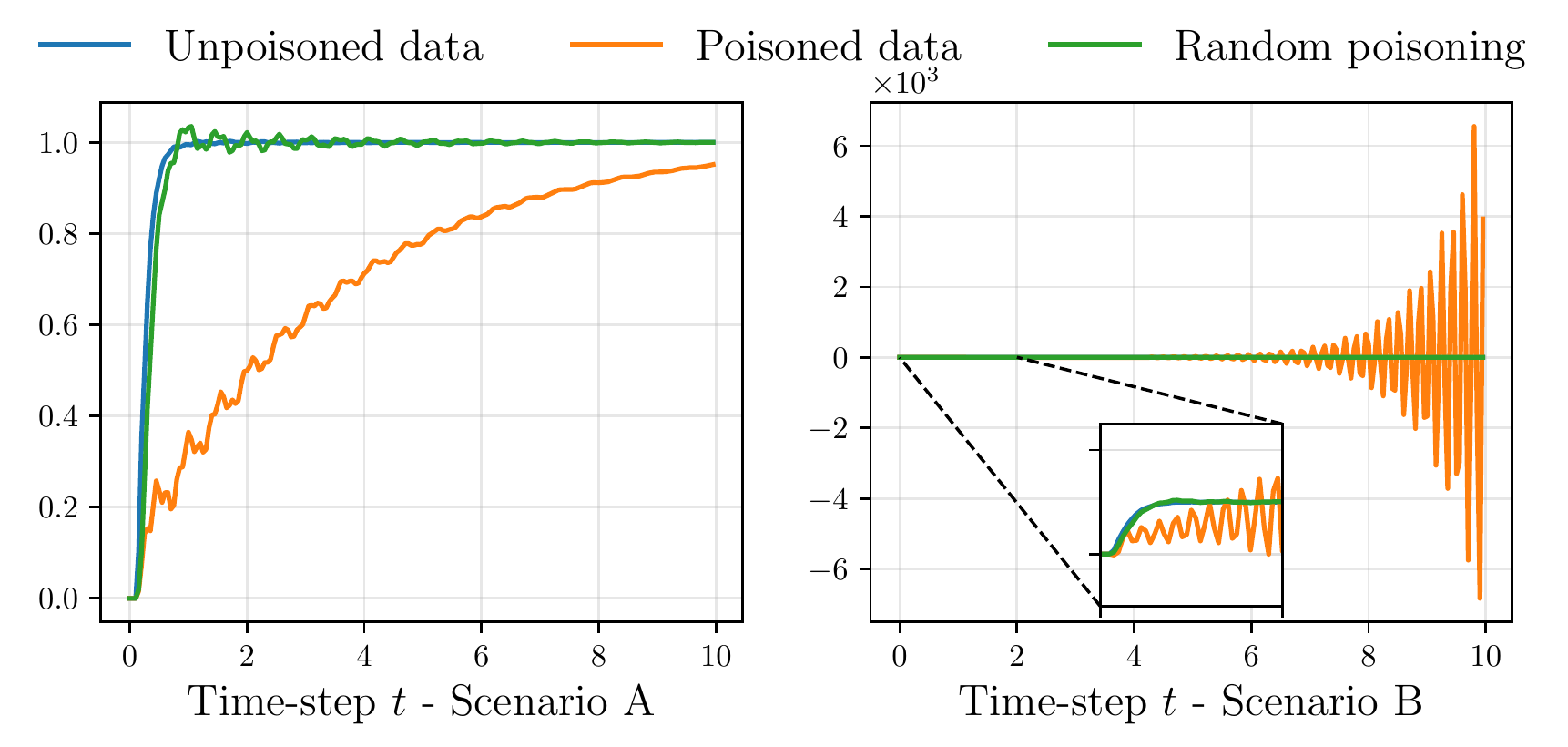}
		\caption{Closed loop system step response for $(\varepsilon_u,\varepsilon_y)=(0.1,0.057)$. }
		\label{fig:step}
	\end{subfigure}
\end{figure*}
As an example, in Figure \oldref{fig:input_output_data_whitenoise}, we show data from a simulation of scenario (B) for the case $\varepsilon_u=0.1,\varepsilon_y=0.057$. We also compare the attacks with a white noise poisoning signal constrained using a fixed $\delta_u$ with $\varepsilon_u=0.2$.

\medskip
\noindent
\textbf{Attack analysis.}
In Figures \oldref{fig:loss}-\oldref{fig:step}, we show the average learner's loss for the various values of $(\varepsilon_u,\varepsilon_y)$ and the step response, for both scenarios (A) and (B). We also show the loss value in the case of a random attack where each attack point is white noise (WN) with constraint $\varepsilon_u=0.2$. In Figure \oldref{fig:loss}, due to the randomness in scenario (B) and non-convexity with respect to $\boldsymbol{a}_y$, we have taken the average value of the loss over $512$ simulations (the shadowed area, though very small, displays the $95\%$ confidence interval). 

First, we observe from Figures \oldref{fig:loss}-\oldref{fig:step} that randomly attacking the dataset does not perturb performance in a significant way. The loss due to random Gaussian poisoning is significantly small compared to the loss due to the poisoning attacks. This indeed shows that our algorithm is exploiting the dynamics of the system in order to compute a poisoning attack.

Then, we observe that the loss in scenario (A) seems to increase more quickly compared to that in scenario (B), but this is not an indication of poor performances. Despite a flat increase of the learner's loss in scenario (B), which may be explained by the usage of persistently exciting data, we observe that a successful poisoning attack may cause more damage in scenario (B).  
\begin{table}[H]
\begin{tabular}{lccccc}
\hline
\rowcolor{Gray}
\multicolumn{1}{l|}{\textbf{Scenario A - $\varepsilon_y$}} & $0.01$   & $0.023$  & $0.037$  & $0.05$  & $0.057$   \\ \hline
\multicolumn{1}{l|}{$\varepsilon_u=0$}                                                        & $0\%$    & $0\%$    & $0\%$    & $0\%$ & $0\%$ \\
\multicolumn{1}{l|}{$\varepsilon_u=0.1$}                                                     & $0.78\%$    & $0.58\%$    & $13.9\%$ & $1.56\%$ & $0.78\%$ \\
\multicolumn{1}{l|}{$\varepsilon_u=0.2$}                                                      & $47.1\%$ & $29.7\%$ & $19.8\%$ & $1.76\%$ & $1.56\%$ \\ 
\rowcolor{Gray}\hline
\multicolumn{1}{l|}{\textbf{Scenario B - $\varepsilon_y$}}  & $0.01$   & $0.023$  & $0.037$  & $0.05$  & $0.057$   \\ \hline
\multicolumn{1}{l|}{$\varepsilon_u=0$}                                    & $0\%$    & $0\%$    & $0\%$    & $0.2\%$ & $0\%$ \\ 
\multicolumn{1}{l|}{$\varepsilon_u=0.1$}                                  & $0\%$    & $0\%$    & $34.4\%$ & $74.6\%$ & $75.2\%$ \\ 
\multicolumn{1}{l|}{$\varepsilon_u=0.2$}                                  & $0\%$ & $1.75\%$ & $54.7\%$ & $84.2\%$ & $88.9\%$ \\ \hline
\end{tabular}
\caption{Average proportion of unstable closed-loop systems out of $512$ simulations.}
\label{table:stability_wn}
\end{table}
As a matter of fact, from Table \oldref{table:stability_wn}, we note that instability of the closed-loop system does not seem to be correlated with the learner's loss in scenario (A), while it seems to be for scenario (B). This is also due to the choice of objective of the malicious agent, which does not try to directly maximize the closed-loop eigenvalues, but just increase the learner's loss, which is also linked to how informative the data is. This raises the interesting question, of whether using informative data may help the malicious agent in making the closed-loop system unstable.
 

Remarkably, we also observe that for $\delta_u=0$, the learner's loss slowly increases with $\delta_y$. From Corollary \oldref{corollary:minmax_bound}, we can see the presence of a multiplicative term $\|\boldsymbol{y}\|^2 \delta_u$ that  may indicate the possibility of $\set L$  diverging  for $\delta_u$ big enough. This effect can be clearly seen in the right plot of Figure \oldref{fig:loss}, where for $\varepsilon_u=0.1$ and $\varepsilon_u=0.2$ the loss starts to diverge. Moreover, from Table \oldref{table:stability_wn}, we see  that for $\varepsilon_u=0$, the closed-loop system never gets unstable. This suggests that robustness of data-driven control methods can be assessed by a strict analysis of the input data integrity.

%% file: 5.conclusion.tex
\section{Conclusion}
In this work, we have introduced a generic bi-level optimization objective that can be used to compute attacks on data-driven control methods. We then specialized the method to the well known Virtual Reference Feedback Tuning technique. In general, this bi-level optimization problem is non-convex, whilst in the case of VRFT, it becomes convex in the input data. For the case of attacks against VRFT, we have provided an upper bound on the difference of the poisoned/unpoisoned control law parameters and introduced a min-max objective that the adversary can use to compute an attack. Our analysis and experiments have shown that the usage of exciting signals may emphasize the effect of the poisoning attack and that a strict integrity check of the input data may reduce the impact of the attack. Future analysis should focus on a more in-depth theoretical analysis of the attack and find additional defence strategies for the learner.

%% file: 6.appendix_vrft_targeted.tex

\section{Additional attacks on VRFT}
\subsection{Targeted attack}
As mentioned in the previous section, the max-min objective does not let the attacker choose a particular behaviour of the resulting closed-loop system. Furthermore, corollary \ref{corollary:targeted_theta} shows a way to perform a targeted attack. 
A malicious agent may wish to perform a more sophisticated attack, and try to poison the dataset so that the closed-loop behaviour follows a certain desired reference model. The attacker can encapsulate this desired behavior in a transfer function $M_a(z)$ and perform the VRFT method to find the corresponding vector $\theta_a$ (observe that the control architecture does not change and the malicious agent is  constrained to the set of reference models $\left\{\frac{K_\theta(z)G(z)}{1+K_\theta(z)G(z)}\right\}$). Then, the attack is simply minimizing the distance between $\theta_a$ and $\hat{\theta}$:
\begin{equation}\label{eq:op_vrft_targeted}
\begin{aligned}
\min_{\boldsymbol{u}' \boldsymbol{y}'}\quad &\set A(\boldsymbol{u},\boldsymbol{y}, \hat \theta(\boldsymbol{u}', \boldsymbol{y}'))\coloneqq \frac{1}{n_k}\|\theta_a-\hat{\theta}(\boldsymbol{u}',\boldsymbol{y}')\|_2^2& \\
\textrm{s.t.} \quad & \hat{\theta}(\boldsymbol{u}', \boldsymbol{y}') = \left(\Phi^{\top}(\boldsymbol{y}') \Phi(\boldsymbol{y}') \right)^{-1}\Phi^{\top}(\boldsymbol{y}') \boldsymbol{u}'\\
&\|\boldsymbol{u}'-\boldsymbol{u}\|_{q_u} \leq \delta_u,\quad \|\boldsymbol{y}'-\boldsymbol{y}\|_{q_y} \leq \delta_y
\end{aligned}
\end{equation}
 This objective in a sense mimics  the original VRFT objective in which the attacker tries to minimize criterion \ref{eq:criterion_J} $\left\|M_a(z)-[(I-M_a)GK_\theta](z)\right\|_2^2$ with respect to the new reference model $M_a(z)$. Also, as pointed out in corollary \ref{corollary:targeted_theta} for $\delta_u$ sufficiently large the attacker can choose any $\hat\theta \in \mathbb{R}^{n_k}$.
\begin{lemma}
For a fixed vector $\boldsymbol{y}'$ the targeted attack is convex in $\boldsymbol{a}_u$, and the optimal poisoning signal $\boldsymbol{a}_u$ is
\begin{equation}
\boldsymbol{a}_u =\begin{cases}
\left(P^{\top}P\right)^{\dagger}P^{\top}\theta_a-\boldsymbol{u} \qquad \hbox{ if } \footnotesize{\|\left(P^{\top}P\right)^{\dagger}P^{\top}\theta_a-\boldsymbol{u}\|_{q_u} \leq \delta_u}\\ \frac{\left(P^{\top}P\right)^{\dagger}P\theta_a-\boldsymbol{u}}{\|\left(P^{\top}P\right)^{\dagger}P^{\top}\theta_a-\boldsymbol{u}\|_{q_u}}\delta_u \quad \hbox{otherwise}
\end{cases}
\end{equation}
where $P= ( \Phi^{\top}(\boldsymbol{y}')\Phi(\boldsymbol{y}'))^{-1}\Phi^{\top}(\boldsymbol{y}')$.
\end{lemma}
\begin{proof}
As in the max-min attack one can easily recognize convexity of the criterion in $\boldsymbol{u}'$. Since we  are minimizing a convex problem, we can directly find the solution by solving $\nabla_{\boldsymbol{a}_u} \set J=0$ ( and then normalize it to enforce the norm constraint). Let $\Phi' \coloneqq \Phi(\boldsymbol{y}')$  and $\Phi = \Phi(\boldsymbol{y})$, then we can compute the gradients with respect to $\boldsymbol{a}_u$  to obtain
\begin{equation}0=\nabla_{\boldsymbol{a}_u} \set A = -\frac{2}{N}\Phi'( \Phi^{'T}\Phi')^{-1}(\theta_a- ( \Phi^{'T}\Phi')^{-1}\Phi^{'T}\boldsymbol{u}').
\end{equation}
By letting $P= ( \Phi^{'T}\Phi')^{-1}\Phi^{'T}$ we easily get the result.
\end{proof}
Unfortunately the problem is still non-convex in $\boldsymbol{a}_y$. To compute the attack, one can use the algorithm proposed for the max-min objective \ref{algo1} to compute it. Instead of performing gradient ascent the algorithm performs of gradient descent, and computes the optimal input poisoning signal according to the previous lemma.\\
\renewcommand{\thefigure}{4}
\begin{figure}[t]
	\centering
		\includegraphics[width=\linewidth]{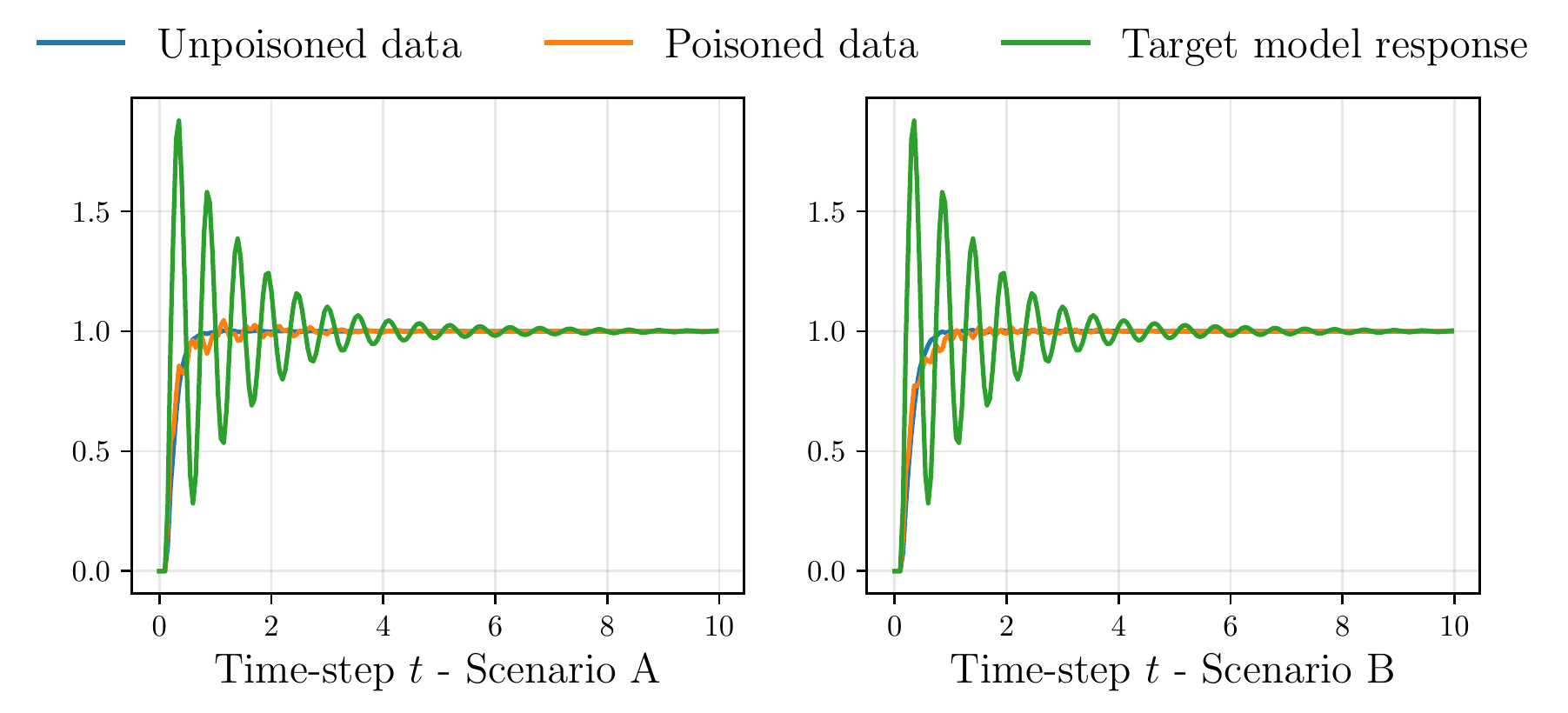}
		\caption{Closed loop step response for $\varepsilon_u=0.1$ and $\varepsilon_y=0.057$.}
		\label{fig:targeted_step}
\end{figure}

\textbf{Simulations.} We analyze the targeted criterion using the same settings as for the max-min criterion, and averaged results over $256$ simulations. As a target model we constructed one using $\theta_a=3\theta_0$, where $\theta_0$ is the solution to the original VRFT problem, taken from \cite{campi2002virtual}, which is $\theta_0\approx \begin{bmatrix}0.33& -0.61 & 0.72& -0.66& 0.48& -0.13\end{bmatrix}$. Future work could better focus on the analysis of choosing $\theta_a$, whilst here we present the idea and some numerical results. \\
\renewcommand{\thefigure}{5}
\begin{figure}[t]
	\centering
		\includegraphics[width=\linewidth]{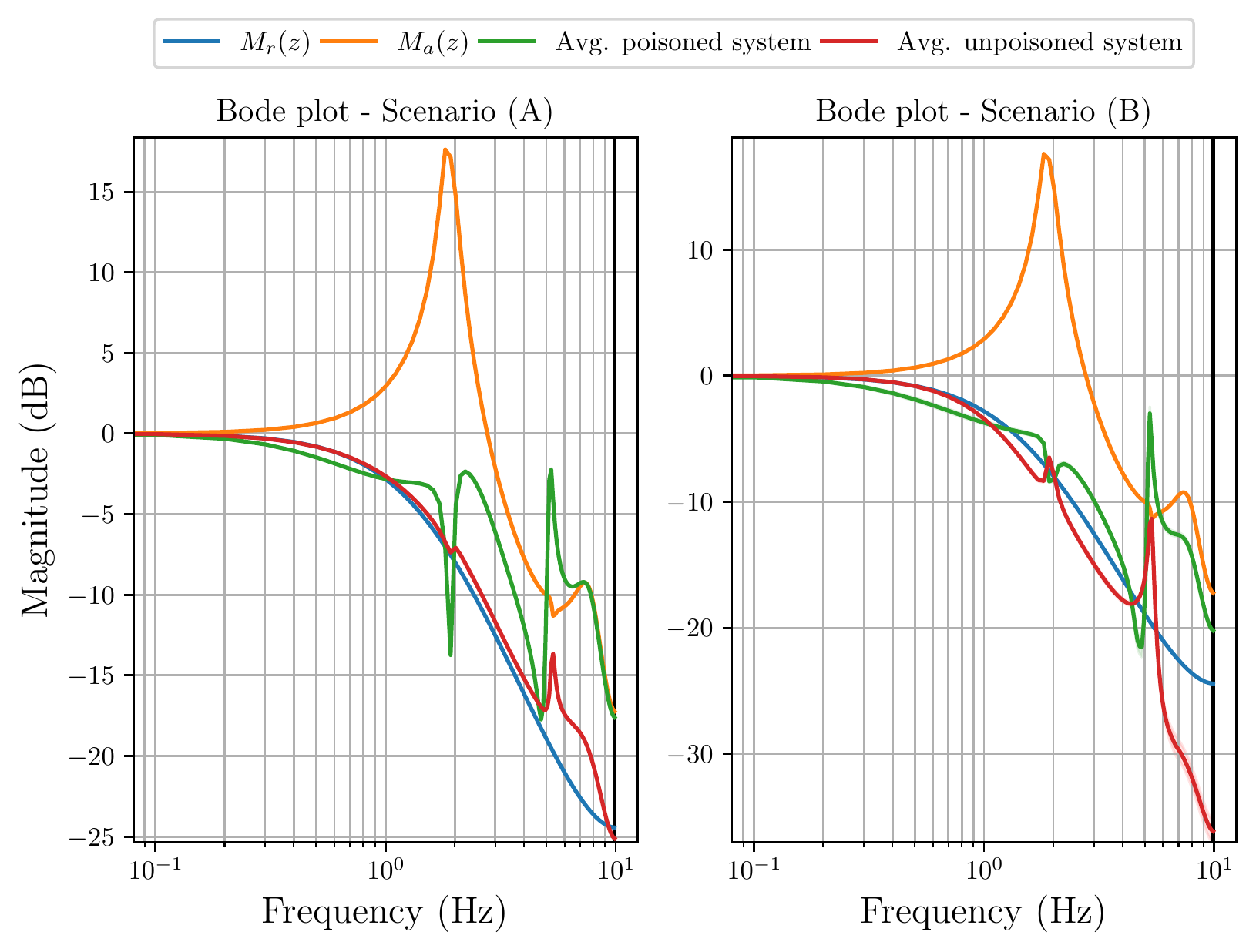}
		\caption{Closed loop Bode response: in blue the reference model; in red the average unpoisoned closed-loop system; in green the poisoned closed-loop system; in orange the target model $M_a(z)$. Averaged curves also display $95\%$ confidence interval.}
		\label{fig:targeted_bode}
\end{figure}

\renewcommand{\thefigure}{6}
\begin{figure}[b]
	\centering
		\includegraphics[width=1\linewidth]{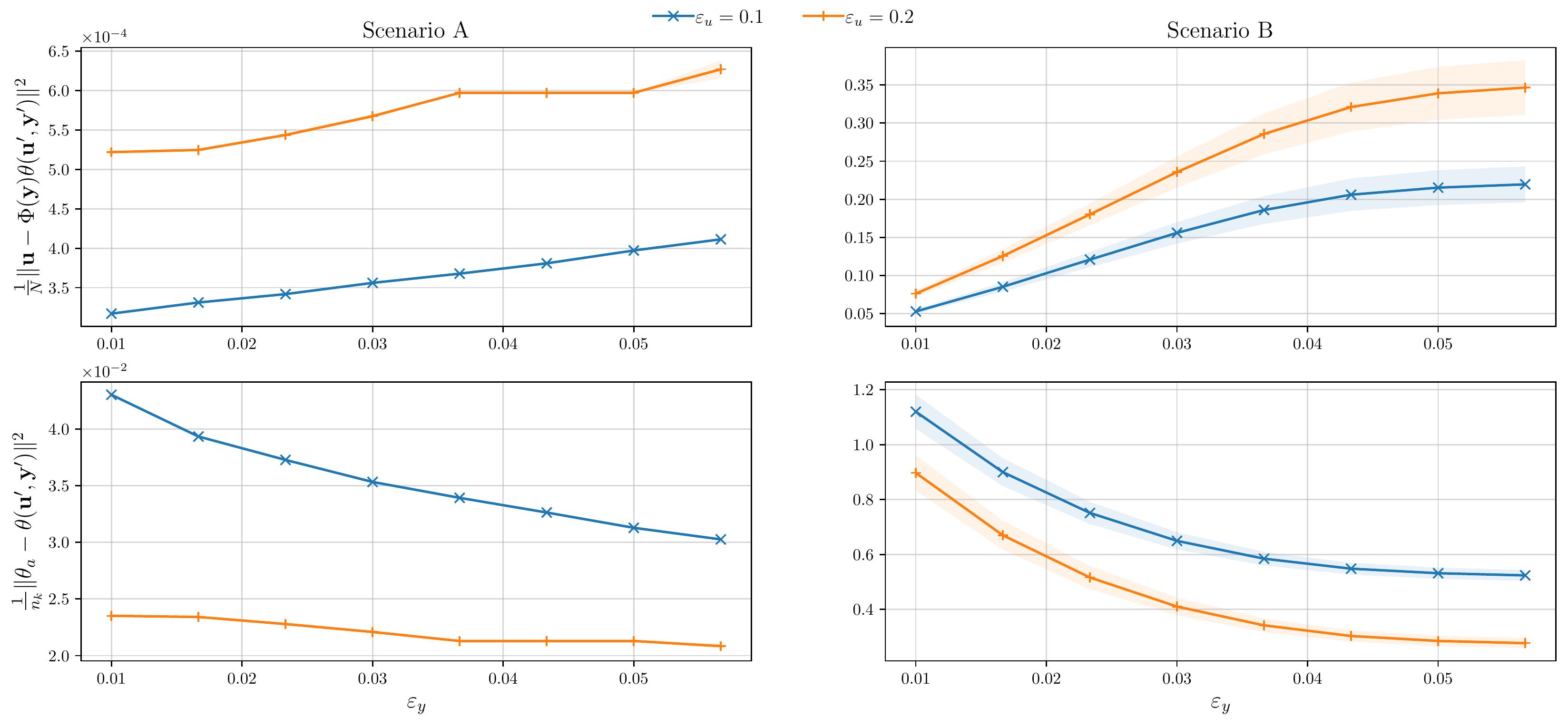}
		\caption{(Top) Loss of the learner $\set L$; (Bottom) Loss of the malicious agent $\set A$. Shadowed area displays $95\%$ confidence interval.}
		\label{fig:targeted_loss}
\end{figure}

In Figures \ref{fig:targeted_step} and \ref{fig:targeted_bode} we present respectively the step responses and bode diagrams for the different models. In Figure \ref{fig:targeted_step} the effect of the target model $M_a(z)$ is more visible than the other curves, whilst the poisoned closed-loop system seems to be barely affected by the attack. More insights can be gained by looking at the bode diagrams in Figures \ref{fig:targeted_bode}: first we observe that the poisoned closed-loop system better approximates $M_a(z)$ only at low and high frequencies, but not  in the range of frequencies where there is a peak. Furthermore, this approximation gets better in the case of scenario (A). This is in line with the reasoning provided in Proposition \ref{corollary:targeted_theta}. This also motivates a possible future work of analyzing which values of $\theta_a$ are more feasible for the malicious agent.\\

Finally, in Figure \ref{fig:targeted_loss} we show respectively: in the top row the learner's loss $\set  L$ for scenario (A) and (B); in the bottom row the loss $\set A$ of the malicious agent. Despite a rapid convergence of $\theta$ to $\theta_a$, in roughly both scenarios, this convergence is not also depicted in the bode plots. Therefore, we argue that a successful targeted attack needs to take into consideration two things: (i) feasibility of $\theta_a$ and (ii) a good choice of the malicious agent criterion (for example one that considers the frequency domain).

%% file: 7.appendix_willems.tex
\section{Attacks on methods based on Willems lemma}
\subsection{Introduction}
The usage of Willems lemma \cite{willems2005note} inspired the development of new data-driven techniques. For example, in \cite{de2019formulas} they show how to derive data-driven solutions of control problem by using linear matrix inequalities in conjunction with Willems lemma. Important examples are optimal control problems.  In \cite{de2019formulas} Theorem 1,  the authors show that for the system $x_{t+1}=Ax_t+Bu_t$, we can equivalently write
\[ \begin{bmatrix}
B & A
\end{bmatrix}= X_{[1,N]}\begin{bmatrix}
U_{[N-1]}\\X_{[N-1]}
\end{bmatrix}^{\dagger}\]
where $X$ and $U$ represent data collected from the system and $\dagger$ denotes the right inverse. Observe that the above representation holds only if the input sequence is an exciting input of order $n+1$ and $\textrm{rank}\begin{bmatrix}U_{[N-1]} & X_{[N-1]}\end{bmatrix}=n+m$ \ifdefined\isextended\else (please refer to \ref{techrepor} for a definition of persistently exciting signal)\fi. Furthermore, as shown in Theorem 2 in \cite{de2019formulas}, any closed-loop system with a state-feedback control $u_t=Kx_t$ we have the following equivalent representation:
\begin{equation}x_{t+1}=X_{[1,N]}^\top G_Kx_t\end{equation}
where $G_K$ is a $T\times n$ matrix satisfying
\begin{equation}K=U_{[N-1]}^\top G_K \hbox{ and } I_n=X_{[N-1]}^\top  G_K.\end{equation}
We can equivalently write  $A+BK=X_{[1,N]}^\top G_K$: this allows to treat $G_K$ as a decision variable, and search for a matrix $G_K$ that satisfies some  performance conditions \cite{de2019formulas}.

\subsection{Optimal control design using Willems lemma}
As mentioned, one can consider the problem of finding the optimal control law using Willems lemma. We will now revise the technique shown in \cite{de2019formulas}. Consider the following system
\[x_{t+1}=Ax_t+Bu_t+\xi_t\]
where $x_t\in \mathbb{R}^n, u_t\in \mathbb{R}^m$ and $\xi_t$ is an external input to the system. The objective is to find a linear state feedback control law $u_t=Kx_t$ that minimizes the influence of the external input onto $x$ with $K\in \mathbb{R}^{m\times n}$. For that purpose, let $z_t = Q_x^{1/2} x_t$ be a performance metric signal with $Q_x \in \mathbb{R}^{n\times n}$ and $Q\succeq 0$. Define $H_K(z) \coloneqq Q_x(zI-A-BK)^{-1}$: the transfer function from $\xi$ to $z$. Then the learner wishes to learn $K$ that minimizes the $H_2$ norm of $H_K(z)$:
\[\|H_K(z)\|_2^2 = \frac{1}{2\pi}\int_{-\pi}^{\pi}\Tr\left[H_K(e^{j\omega}) H_K^{\top}(e^{-j\omega})\right]\textrm{d}\omega.\]
Minimizing $H$ is equivalent to minimizing  the mean-square deviation of $z$ when $\xi$ is a white process with unit covariance.
In \cite{de2019formulas} the authors show that optimal controller $K$ can be found using the equation $K=U_{[N-1]}Q(X_{[N-1]}Q)^{-1}$ where $Q \in \mathbb{R}^{N\times n}$ is solution of
\begin{equation}\label{eq:willems_control_h2}
\begin{aligned}
\min_{Q}\quad & \Tr\left(Q_x X_{[N-1]}Q\right)& \\
\textrm{s.t.} \quad & \begin{bmatrix}
X_{[N-1]}Q -I_n & X_{[1,N]}Q\\[0.5
em]
Q^{\top} X_{[1,N]}^{\top} & X_{[N-1]}Q
\end{bmatrix}\succeq 0.
\end{aligned}
\end{equation}
\subsection{Attack formulation}
Based on problem \ref{eq:willems_control_h2} we will now devise an attack strategy for the malicious agent. Let the poisoning signals be $A_x\in \mathbb{R}^{n\times (N+1)}$ and $A_u \in \mathbb{R}^{m \times N}$. Define the poisoned data matrices $X_{[N-1]}'\coloneqq X_{[N-1]}+A_{x,[N-1]}$ and $X_{[1,N]}'\coloneqq X_{[1,N]}+A_{x,[1,N]}$ (similarly define $U_{[N-1]}'$). Let
\[
M\left(X_{[N]}'\right)\coloneqq\begin{bmatrix}
X_{[N-1]}'Q -I_n & X_{[1,N]}'Q\\[0.5
em]
Q^{\top} X_{[1,N]}^{'T} & X_{[N-1]}'Q
\end{bmatrix}\succeq 0
\]
where the notation $X'$ indicates the poisoned version of the signals.\\
 \textbf{Attack problem. } Then we can cast the problem of finding the poisoning attack for a genetic attack criterion $\set A$ as follows
\begin{equation}\label{eq:willems_attack_general}
\begin{aligned}
&\max_{X_{[N]}', U_{[N-1]}'}\quad  \set A(\set D_N, K)& \\
\textrm{s.t.} \quad & Q = \argmin_{Q} \left\{\Tr\left(Q_x X_{[N-1]}'Q\right): M\left(X_{[N]}'\right) \succeq 0\right\}\\
& K = U_{[N-1]}' Q (X_{[N-1]}' Q)^{-1}\\
& \|X_{[N]}'- X_{[N]}\|_{q_x} \leq \delta_x\\ & \|U_{[N-1]}'-U_{[N-1]}\|_{q_u} \leq \delta_u
\end{aligned}
\end{equation}
where $q_x,q_u$ are convex norms. To solve it, since the inner problem is convex and regular, we can resort again to the Single-Level Reduction method \cite{sinha2017review,bard2013practical}, and replace the inner problem with its KKT conditions.  We get the following:

\begin{lemma}
Denote by $\mathcal{S}^n$ the space of $n\times n$ symmetric matrices  equipped with the inner product $\langle A,B \rangle=\Tr(A^{\top}B)$. Then problem \ref{eq:willems_attack_general} is equivalent to the following problem:
\begin{equation}\label{eq:willems_attack_general_dual}
\begin{aligned}
&\max_{X_{[N]}', U_{[N-1]}',Q,Z}\quad \set A(\set D_N, K)& \\
\textrm{s.t.} \quad &  Q_x X_{[N-1]}'-(Z_{11}+ Z_{22})^{\top}X_{[N-1]}'-2Z_{12}^{\top}X_{[1,N]}'=0\\
& M(X_{[N]}') \succeq 0,\quad \langle Z, M(X_{[N]}') \rangle = 0,\quad Z\succeq 0\\
& K = U_{[N-1]}' Q (X_{[N-1]}' Q)^{-1}\\
& \|X_{[N]}'- X_{[N]}\|_{q_x} \leq \delta_x\\ & \|U_{[N-1]}'-U_{[N-1]}\|_{q_u} \leq \delta_u
\end{aligned}
\end{equation}
with $Q\in \mathbb{R}^{N\times n}$ and $Z\in \set S^{2n}$.
\end{lemma}
\begin{proof}
Denote by $\mathcal{S}^n$ the space of $n\times n$ symmetric matrices, which is equipped with the inner product $\langle A,B \rangle=\Tr(A^{\top}B)$. We can  we can write the Lagrangian function of the inner problem as follows
\[L(X_{[N]}',Q,Z) = \Tr\left(Q_x X_{[N-1]}'Q\right) -\langle Z, M(X_{[N]}') \rangle,  \]
where $Z\in \set S^{2n}$. Express $Z$ in block form
\[Z=\begin{bmatrix}
Z_{11} & Z_{12}\\
Z_{12}^{\top} & Z_{22}
\end{bmatrix}\]
then $L$ admits the following rewriting
\begin{align*}
L(X_{[N]}',Q,Z) &= \langle Q_x^{\top}, X_{[N-1]}'Q\rangle -\langle Z, M(X_{[N]}') \rangle\\
&=\langle Q_x^{\top}, X_{[N-1]}'Q\rangle  -\langle Z_{11},X_{[N-1]}'Q -I_n \rangle\\
&\qquad -\langle Z_{22},X_{[N-1]}'Q \rangle - 2\langle Z_{12},X_{[1,N]}'Q\rangle\\
&=\langle Q_x^{\top}, X_{[N-1]}'Q\rangle +\langle Z_{11},I_n \rangle\\
&\qquad -\langle Z_{11}+ Z_{22},X_{[N-1]}'Q \rangle \\
&\qquad - 2\langle Z_{12},X_{[1,N]}'Q\rangle
\end{align*}
from which follows the first KKT condition:
\[0=\frac{\textrm{d}L}{\textrm{d}Q} = Q_x X_{[N-1]}'-(Z_{11}+ Z_{22})^{\top}X_{[N-1]}'-2Z_{12}^{\top}X_{[1,N]}'.\]
The remaining KKT conditions are $Z\succeq 0, M(X_{[N]}') \succeq 0$ and $ \langle Z, M(X_{[N]}') \rangle = 0$, from which we deduce the result.\\
\end{proof}

Unfortunately, the problem is hard to solve, not only because of the upper-level constraint on $Q$, but also because the Lagrangian constraints lead to non-convexities. The complementary conditions can  be thought of as a boolean variables, turning the problem into a mixed-integer problem \cite{sinha2017review,sinha2017evolutionary,dempe2015bilevel,dempe2019solution}. In general, one can try to use branch and bounds methods or evolutionary algorithms \cite{sinha2017review} (the latter approach has shown good performance in general for bilevel optimization). Alternatively, one can also consider solving a sequence of relaxed problems, as shown in \cite{mersha2008solution}. We will now provide a few examples of attacks.
\subsection{Max-min attack}
The first attack we propose is the max-min attack. This attack is implicitly maximizing the $H_2$ norm of $H_K(z)$, making the system more susceptible to modeling errors and process noise. To formulate it, one must be careful and observe that the malicious agent cannot merely choose $\set A(\set D_N, K)= \Tr(Q_x X_{[N-1]}Q)$. The reason is simple: the matrix $Q$ does not parametrize the system $(A,B)$ because of the dependence of $Q$ on the attack signal $A_x$. In order to have a correct parametrization, we need to enforce the constraints
\[ \begin{bmatrix}
K\\
 I_n
\end{bmatrix} = \begin{bmatrix}
U_{[N]}\\
X_{[N]}
\end{bmatrix}G_K,\]
where  $G_K$ is a decision variable and $K$ is the poisoned control law. We can derive the following attack:
\begin{lemma}\label{lemma:maxmin_attack}
Consider the optimization problem in \ref{eq:willems_attack_general}, then the optimal attack maximizing the $H_2$ norm of $H_K(z)$ can be found  by solving the optimization problem
\begin{equation}
\begin{aligned}
&\max_{X_{[N]}', U_{[N-1]}',W,G_K}\quad  \Phi(Q_xW)& \\
\textrm{s.t.} \quad & Q = \argmin_{Q} \left\{\Tr\left(Q_x X_{[N-1]}'Q\right): M(A_x) \succ 0\right\}\\
& X_{[1,N]	}G_K W G_K^{\top} X_{[1,N]}^{\top} -W+I_n= 0,\quad W\succeq I_n \\
& U_{[N-1]}G_K =  U_{[N-1]}' Q (X_{[N-1]}' Q)^{-1}\\
& X_{[N-1]}G_K = I_n\\
& \|X_{[N]}'- X_{[N]}\|_{q_x} \leq \delta_x\\ & \|U_{[N-1]}'-U_{[N-1]}\|_{q_u} \leq \delta_u
\end{aligned}
\end{equation}
 where $W \in \set S^{n}$ and $\Phi$ satisfies
 \[\Phi(Q_xW) = \begin{cases}
 \Tr(Q_xW) \hbox{ if } \max_i|\lambda_i(X_{[1,N]}G_K)| \leq 1\\
 \infty
 \end{cases}\] \end{lemma}
 \begin{proof}
 The result follows from the fact that maximizing the $H_2$ norm of $H_K(z)$ is equivalent to maximizing the trace of $Q_x W$ where $W$ is the controllability Gramian of the closed-loop system $A+BK$. The matrix $W$ for a stable closed-loop system can be found by solving
 \[(A+BK)W(A+BK)^{\top}-W+I_n=0\]
 for a symmetric matrix $W\succeq I_n$.
 Due to Theorem 2 in \cite{de2019formulas} we have $A+BK=X_{[1,N]}G_K$ and  
\[ \begin{bmatrix}
K\\
 I_n
\end{bmatrix} = \begin{bmatrix}
U_{[N]}\\
X_{[N]}
\end{bmatrix}G_K.\] Then the result follows by observing that   $K= U_{[N]}' Q (X_{[N]}' Q)^{-1}$. In case $A+BK$ is unstable, then thematrix  $W$ cannot be computed, from which follows the definition of $\Phi$.\\
\end{proof}

The above formulation of the attack problem has  several limitations, mainly due to the computation of the Gramian matrix of $A+BK$. We believe instead that an alternative formulation of the problem, easier to compute, is the overall maximization of the closed-loop eigenvalues.
\subsection{Eigenvalues attack}
To maximize the closed-loop eigenvalues one can use the fact that the following equality holds
\begin{align*}\det((A+BK)(A+BK)^{\top})= \prod_{i=1}^n |\lambda_i(A+BK)|^2\end{align*}
Therefore maximizing $\det((A+BK)(A+BK)^{\top})$ can be used as a proxy to maximize the absolute value of the closed-loop eigenvalues. For simplicity, we will consider the  concave criterion $\log(\det((A+BK)(A+BK)^{\top}))$ (note that one could alternatively maximize the geometric mean of the squared eigenvalues $\sqrt[n]{\det((A+BK)(A+BK)^{\top}}$, which is a concave function). Using Theorem 2 in \cite{de2019formulas} one can  write $\log(\det(X_{[1,N]}G_KG_K^{\top} X_{[1,N]}^{\top}) )$ with $G_K$ satisfying 
\[ \begin{bmatrix}
K\\
 I_n
\end{bmatrix} = \begin{bmatrix}
U_{[N]}\\
X_{[N]}
\end{bmatrix}G_K\] 
and $K=U_{[N]}' Q (X_{[N]}' Q)^{-1}$ (where $Q$ is solution of the inner problem). Therefore we obtain the following optimization problem

\begin{lemma}\label{lemma:eigenvalues_attack}
Consider the optimization problem in \ref{eq:willems_attack_general}, then we can maximize the closed-loop eigenvalues by solving the following optimization problem
\begin{equation}
\begin{aligned}
&\max_{X_{[N]}', U_{[N-1]}', G_K}\quad  \log(\det(X_{[1,N]}G_KG_K^{\top} X_{[1,N]}^{\top})) & \\
\textrm{s.t.} \quad & Q = \argmin_{Q} \left\{\Tr\left(Q_x X_{[N-1]}'Q\right): M(X_{[N]}') \succeq 0\right\}\\
& U_{[N-1]}G_K =  U_{[N-1]}' Q (X_{[N-1]}' Q)^{-1}\\
& X_{[N-1]}G_K = I_n\\
& \|X_{[N]}'- X_{[N]}\|_{q_x} \leq \delta_x\\ & \|U_{[N-1]}'-U_{[N-1]}\|_{q_u} \leq \delta_u
\end{aligned}
\end{equation} \end{lemma}
Unfortunately, the parametrization introduced by $G_K$ makes the problem particularly difficult, and non-convex. Not only we have the non-convexity introduced by using the single-level reduction method, but also the non-convex constraint $U_{[N-1]}G_K =  U_{[N-1]}' Q (X_{[N-1]}' Q)^{-1}$. Despite that, we can still rewrite the problem in order to devise an attack on the input signal.
By using the relationship
\[ \begin{bmatrix}
B & A
\end{bmatrix}= X_{[1,N]}\begin{bmatrix}
U_{[N-1]}\\X_{[N-1]}
\end{bmatrix}^{\dagger}\]
and the fact that $A+BK= X_{[1,N]}G_K$ and $U_{[N-1]}G_K=K$ we can discard the variable $G_K$ and  write
\begin{lemma}
\begin{equation}
\begin{aligned}
&\max_{X_{[N]}', U_{[N-1]}'}\quad  \log(\det((A+BK)(A+BK)^\top) & \\
\textrm{s.t.} \quad & Q = \argmin_{Q} \left\{\Tr\left(Q_x X_{[N-1]}'Q\right): M(A_x) \succeq 0\right\}\\
& K =  U_{[N-1]}' Q (X_{[N-1]}' Q)^{-1}\\
& \|A_{x}\|_{q_x} \leq \delta_x, \quad\|A_{u}\|_{q_u} \leq \delta_u
\end{aligned}
\end{equation}
\end{lemma}
which, as a corollary, allows  us to formalize an attack on the input signal:

\begin{corollary}\label{eq:lemma_optimal_input_attack_willems}
Consider a fixed attack vector $A_x$ and let $Q = \argmin_{Q} \left\{\Tr\left(Q_x X_{[N-1]}'Q\right): M(X_{[N]}') \succeq 0\right\}$. Then, the optimization problem that finds the optimal poisoning on the input signal is
\begin{equation}
\begin{aligned}
\max_{U_{[N-1]'}}\quad & \log(\det((A+BK)(A+BK)^{\top})) & \\
\textrm{s.t.} \quad & K =  U_{[N-1]}' Q (X_{[N-1]}' Q)^{-1}\\
&\|U_{[N-1]}'-U_{[N-1]}\|_{q_u} \leq \delta_u
\end{aligned}
\end{equation}
\end{corollary}
The previous problem 
is not easy to solve, but easier than the previous problems that we showed. We will now present some numerical results.
\subsection{Eigenvalues attack - Simulations}
\textbf{Plant dynamics and optimal control.} In this section, we present some numerical results of attacking control methods designed using Willems lemma. We  use the same plant being analyzed in \cite{de2019formulas}, which is  
the discretized version of a batch reactor system \cite{walsh2001scheduling} with a sampling time of $T_s=0.1[s]$. The dynamics of the plant are given by
\begin{equation*}\resizebox{\hsize}{!}{%
$A=\begin{bmatrix}
1.178 & 0.001 & 0.511 & -0.403\\
-0.051 & 0.661 & -0.011 & 0.061\\
0.076 & 0.335 & 0.560 & 0.382\\
0 & 0.335 & 0.089 & 0.849
\end{bmatrix},\quad B=\begin{bmatrix}
0.004 & -0.087\\
0.467 & 0.001 \\
0.213 & -0.235 \\
0.213 & -0.016
\end{bmatrix}$ %
}\end{equation*}
The optimal control law of this system can be found using Theorem 4 in \cite{de2019formulas}. The input signal used by the learner is white noise $u_t\sim \mathcal{N}(0,1)$, and we run our simulations for $N=15$ and $N=50$, and averaged results over $256$ simulations.\\

\renewcommand{\thefigure}{7}
\begin{figure}[t]
	\centering
		\includegraphics[width=1\linewidth]{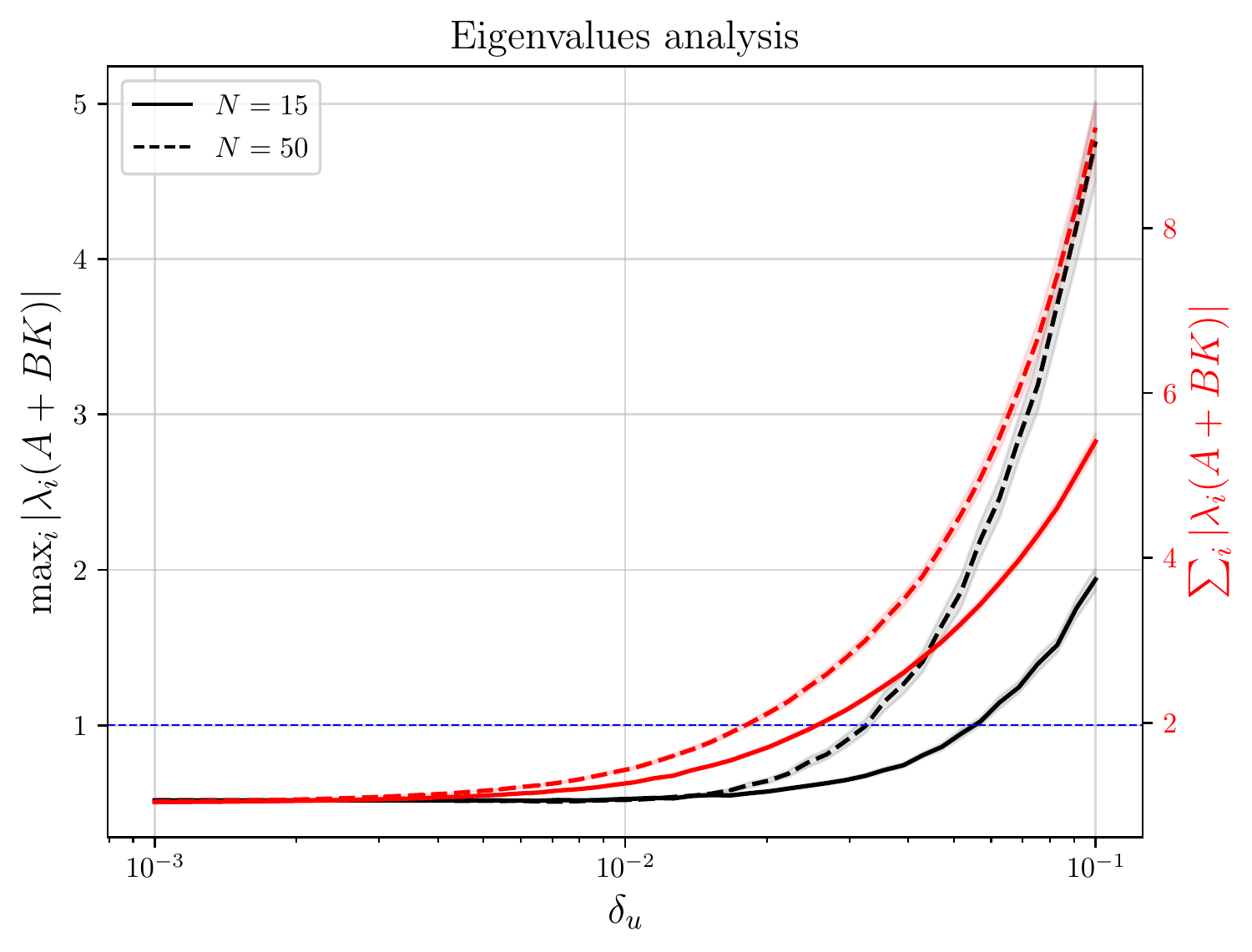}
		\caption{Eigenvalues of the poisoned closed-loop system for different values of $\delta_u$. In red is shown the sum of the absolute value of the closed-loop eigenvalues for $N=15$ and $N=50$. Similarly, in black is shown the absolute value of the maximum eigenvalue of $A+BK$.}
		\label{fig:willems_eigenvalues}
\end{figure}
\renewcommand{\thefigure}{8}
\begin{figure}[b]
	\centering
		\includegraphics[width=1\linewidth]{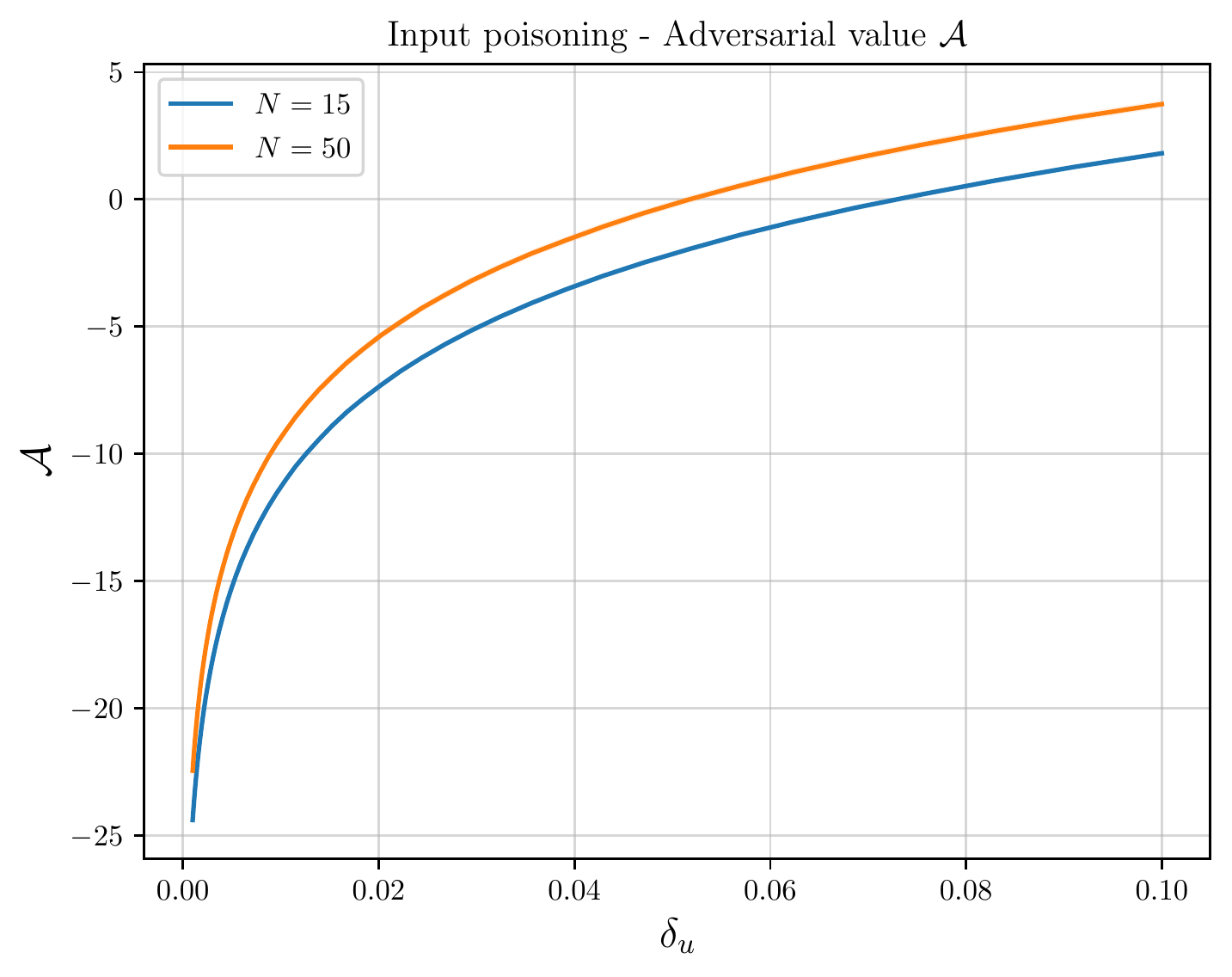}
		\caption{Adversarial objective function $\set A$  for different values of $\delta_u$.}
		\label{fig:willems_loss}
\end{figure}
\renewcommand{\thefigure}{9}
\begin{figure*}[]
	\centering
		\includegraphics[width=1\linewidth]{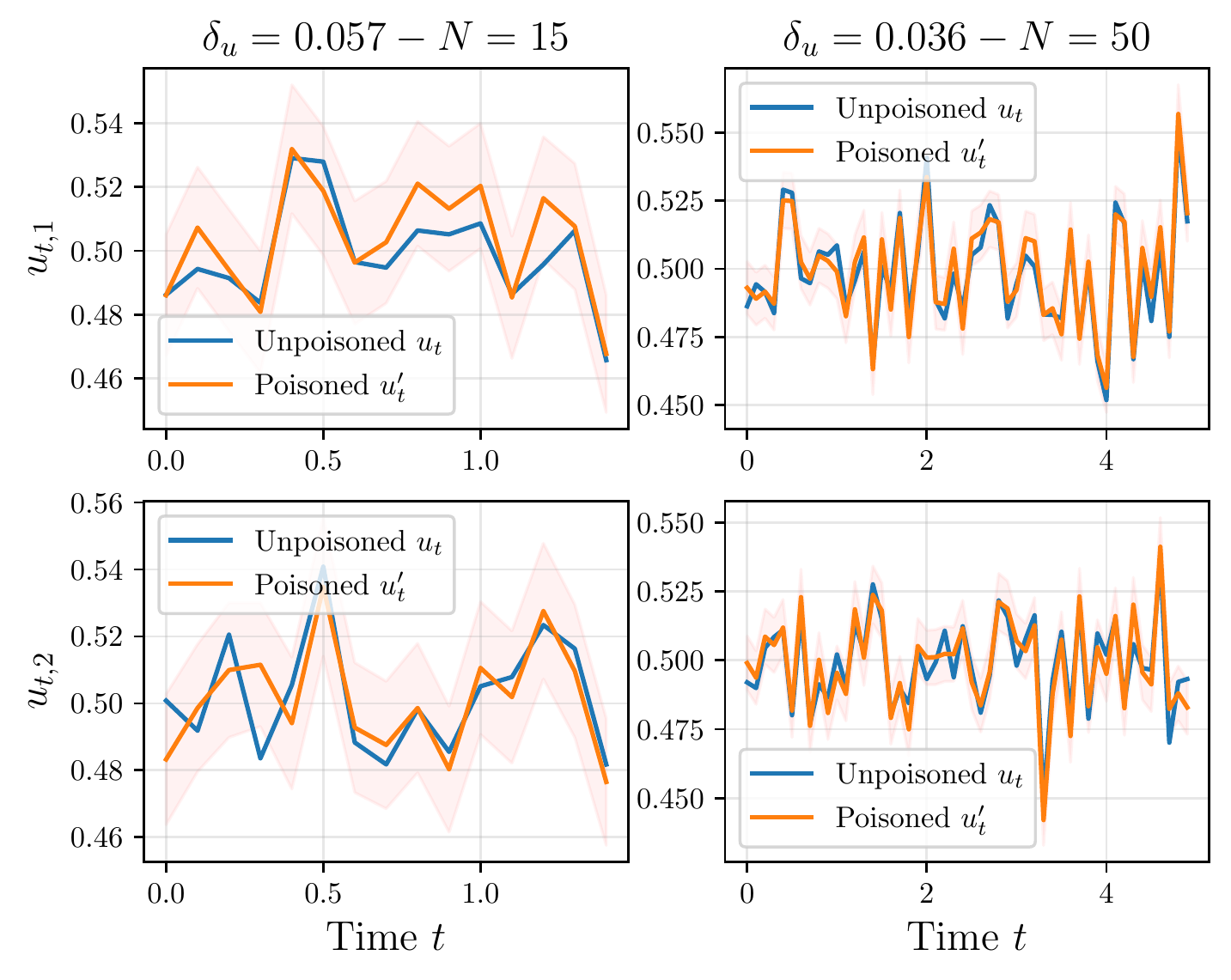}
		\caption{Plots of the input signals for $N=15$ and $N=50$. In each column are shown $u_{t,1}$ and $u_{t,2}$ and the corresponding poisoned signals. The signals are averaged out of $256$ simulations, and the orange shadowed area displays $95\%$ confidence interval for the poisoning signals.}
		\label{fig:willems_signals}
\end{figure*}
\textbf{Attack setting and properties. }For the sake of simplicity, we analyze the problem of poisoning the input signal using the optimization problem \ref{eq:lemma_optimal_input_attack_willems}, solving it using the Particle Swarm Optimizer of MATLAB. We analyze the attack objective for different values of $\delta_u$, using the $\ell_\infty$  norm. As an example, in Figure \ref{fig:willems_signals} we show some examples of poisoned and unpoisoned input signals for which the closed-loop system is unstable. \\


\textbf{Attack analysis.} 
In Figure \ref{fig:willems_eigenvalues} we show some properties of the closed-loop eigenvalues of the closed-loop system. We depict results for $N=15$ (continuous line) and $N=50$ (dashed line). On the left $y$-axis is shown the maximum absolute pole of $A+BK$ whilst on the right $y$-axis it is displayed the sum of the absolute value of the closed-loop poles. A first observation is that increasing the number of points $N$ makes the closed-loop system unstable for smaller values of $\delta_u$. In Figure \ref{fig:willems_loss}, for completeness, we also show the values of the adversarial criterion $\set A$ for different values of $\delta_u$. \\

From simulations, we see that this method is somehow less robust than VRFT. For very small values of $\delta_u$ we can get the system unstable. We have not included results on poisoning the output signals, but, preliminary results show that an attacker can make the closed-loop system unstable even for perturbations computed using the constraint $\delta_y=10^{-3}$ (even without poisoning the initial condition $x_0$). Future analysis should consider a possible way to robustify the method proposed in \cite{de2019formulas} and theoretical analysis of the attack.

%% file: 8.appendix_additional_stuff.tex
\section{Software, code and hardware}
 All experiments were executed on a stationary desktop computer, featuring an Intel Xeon Silver 4110 CPU, 48GB of RAM. Ubuntu 18.04 was installed on the computer. Ubuntu is an open-source Operating System using the Linux kernel and based on Debian. For more information, please check \url{https://ubuntu.com/}.
\\\\
We set up our experiments using Python 3.7.7 \cite{van1995python} (For more information, please refer to the following link \url{http://www.python.org}), and made use of the following libraries:  NumPy version 1.18.1 \cite{oliphant2006guide}, SciPy version 1.4.1 \cite{2020SciPy-NMeth}, PyTorch version 1.4.0 \cite{paszke2017automatic}. All the code will be published on GitHub with MIT license.

%% file: ddcontrol_poisoning_extended.bbl
\begin{thebibliography}{10}
\providecommand{\url}[1]{#1}
\csname url@rmstyle\endcsname
\providecommand{\newblock}{\relax}
\providecommand{\bibinfo}[2]{#2}
\providecommand\BIBentrySTDinterwordspacing{\spaceskip=0pt\relax}
\providecommand\BIBentryALTinterwordstretchfactor{4}
\providecommand\BIBentryALTinterwordspacing{\spaceskip=\fontdimen2\font plus
\BIBentryALTinterwordstretchfactor\fontdimen3\font minus
  \fontdimen4\font\relax}
\providecommand\BIBforeignlanguage[2]{{%
\expandafter\ifx\csname l@#1\endcsname\relax
\typeout{** WARNING: IEEEtran.bst: No hyphenation pattern has been}%
\typeout{** loaded for the language `#1'. Using the pattern for}%
\typeout{** the default language instead.}%
\else
\language=\csname l@#1\endcsname
\fi
#2}}

\bibitem{campi2002virtual}
M.~C. Campi, A.~Lecchini, and S.~M. Savaresi, ``Virtual reference feedback
  tuning: a direct method for the design of feedback controllers,''
  \emph{Automatica}, vol.~38, no.~8, pp. 1337--1346, 2002.

\bibitem{formentin2012non}
S.~Formentin, S.~Savaresi, and L.~Del~Re, ``Non-iterative direct data-driven
  controller tuning for multivariable systems: theory and application,''
  \emph{IET control theory \& applications}, vol.~6, no.~9, pp. 1250--1257,
  2012.

\bibitem{formentin2014comparison}
S.~Formentin, K.~Van~Heusden, and A.~Karimi, ``A comparison of model-based and
  data-driven controller tuning,'' \emph{International Journal of Adaptive
  Control and Signal Processing}, vol.~28, no.~10, pp. 882--897, 2014.

\bibitem{hjalmarsson1998iterative}
H.~Hjalmarsson, M.~Gevers, S.~Gunnarsson, and O.~Lequin, ``Iterative feedback
  tuning: theory and applications,'' \emph{IEEE control systems magazine},
  vol.~18, no.~4, pp. 26--41, 1998.

\bibitem{karimi2004iterative}
A.~Karimi, L.~Mi{\v{s}}kovi{\'c}, and D.~Bonvin, ``Iterative correlation-based
  controller tuning,'' \emph{International journal of adaptive control and
  signal processing}, vol.~18, no.~8, pp. 645--664, 2004.

\bibitem{karimi2007non}
A.~Karimi, K.~Van~Heusden, and D.~Bonvin, ``Non-iterative data-driven
  controller tuning using the correlation approach,'' in \emph{2007 European
  Control Conference (ECC)}.\hskip 1em plus 0.5em minus 0.4em\relax IEEE, 2007,
  pp. 5189--5195.

\bibitem{sala2005extensions}
A.~Sala and A.~Esparza, ``Extensions to “virtual reference feedback tuning: A
  direct method for the design of feedback controllers”,'' \emph{Automatica},
  vol.~41, no.~8, pp. 1473--1476, 2005.

\bibitem{campestrini2011virtual}
L.~Campestrini, D.~Eckhard, M.~Gevers, and A.~S. Bazanella, ``Virtual reference
  feedback tuning for non-minimum phase plants,'' \emph{Automatica}, vol.~47,
  no.~8, pp. 1778--1784, 2011.

\bibitem{lequin2003iterative}
O.~Lequin, M.~Gevers, M.~Mossberg, E.~Bosmans, and L.~Triest, ``Iterative
  feedback tuning of pid parameters: comparison with classical tuning rules,''
  \emph{Control Engineering Practice}, vol.~11, no.~9, pp. 1023--1033, 2003.

\bibitem{de2019formulas}
C.~De~Persis and P.~Tesi, ``Formulas for data-driven control: Stabilization,
  optimality, and robustness,'' \emph{IEEE Transactions on Automatic Control},
  vol.~65, no.~3, pp. 909--924, 2019.

\bibitem{coulson2019data}
J.~Coulson, J.~Lygeros, and F.~D{\"o}rfler, ``Data-enabled predictive control:
  In the shallows of the deepc,'' in \emph{2019 18th European Control
  Conference (ECC)}.\hskip 1em plus 0.5em minus 0.4em\relax IEEE, 2019, pp.
  307--312.

\bibitem{willems2005note}
J.~C. Willems, P.~Rapisarda, I.~Markovsky, and B.~L. De~Moor, ``A note on
  persistency of excitation,'' \emph{Systems \& Control Letters}, vol.~54,
  no.~4, pp. 325--329, 2005.

\bibitem{goodfellow2014explaining}
I.~J. Goodfellow, J.~Shlens, and C.~Szegedy, ``Explaining and harnessing
  adversarial examples,'' \emph{arXiv preprint arXiv:1412.6572}, 2014.

\bibitem{biggio1}
B.~Biggio, B.~Nelson, and P.~Laskov, ``Poisoning attacks against support vector
  machines,'' in \emph{Proceedings of the 29th International Coference on
  International Conference on Machine Learning}, ser. ICML’12.\hskip 1em plus
  0.5em minus 0.4em\relax Madison, WI, USA: Omnipress, 2012, p. 1467–1474.

\bibitem{russo2019optimal}
A.~Russo and A.~Proutiere, ``Optimal attacks on reinforcement learning
  policies,'' \emph{arXiv preprint arXiv:1907.13548}, 2019.

\bibitem{jagielski2018manipulating}
M.~Jagielski, A.~Oprea, B.~Biggio, C.~Liu, C.~Nita-Rotaru, and B.~Li,
  ``Manipulating machine learning: Poisoning attacks and countermeasures for
  regression learning,'' in \emph{2018 IEEE Symposium on Security and Privacy
  (SP)}.\hskip 1em plus 0.5em minus 0.4em\relax IEEE, 2018, pp. 19--35.

\bibitem{munoz2017towards}
L.~Mu{\~n}oz-Gonz{\'a}lez, B.~Biggio, A.~Demontis, A.~Paudice, V.~Wongrassamee,
  E.~C. Lupu, and F.~Roli, ``Towards poisoning of deep learning algorithms with
  back-gradient optimization,'' in \emph{Proceedings of the 10th ACM Workshop
  on Artificial Intelligence and Security}, 2017, pp. 27--38.

\bibitem{chong2019tutorial}
M.~S. Chong, H.~Sandberg, and A.~M. Teixeira, ``A tutorial introduction to
  security and privacy for cyber-physical systems,'' in \emph{2019 18th
  European Control Conference (ECC)}.\hskip 1em plus 0.5em minus 0.4em\relax
  IEEE, 2019, pp. 968--978.

\bibitem{sandberg2015cyberphysical}
H.~Sandberg, S.~Amin, and K.~H. Johansson, ``Cyberphysical security in
  networked control systems: An introduction to the issue,'' \emph{IEEE Control
  Systems Magazine}, vol.~35, no.~1, pp. 20--23, 2015.

\bibitem{mareels1988persistency}
I.~M. Mareels and M.~Gevers, ``Persistency of excitation criteria for linear,
  multivariable, time-varying systems,'' \emph{Mathematics of Control, Signals
  and Systems}, vol.~1, no.~3, pp. 203--226, 1988.

\bibitem{esparza2011neural}
A.~Esparza, A.~Sala, and P.~Albertos, ``Neural networks in virtual reference
  tuning,'' \emph{Engineering Applications of Artificial Intelligence},
  vol.~24, no.~6, pp. 983--995, 2011.

\bibitem{bishop2002art}
M.~A. Bishop, ``Computer security: Art and science,'' 2002.

\bibitem{sinha2017review}
A.~Sinha, P.~Malo, and K.~Deb, ``A review on bilevel optimization: from
  classical to evolutionary approaches and applications,'' \emph{IEEE
  Transactions on Evolutionary Computation}, vol.~22, no.~2, pp. 276--295,
  2017.

\bibitem{techreport}
A.~Russo and A.~Proutiere, ``{Poisoning Attacks against Data-Driven Control
  Methods (extended)},'' \url{https://kth.box.com/v/datapoisoningextended},
  2020.

\bibitem{kavranoglu1993new}
D.~Kavranoglu and M.~Bettayeb, ``A new general state-space representation for
  discrete-time systems,'' \emph{International Journal of Control}, vol.~58,
  no.~1, pp. 33--49, 1993.

\bibitem{shen2016disciplined}
X.~Shen, S.~Diamond, Y.~Gu, and S.~Boyd, ``Disciplined convex-concave
  programming,'' in \emph{2016 IEEE 55th Conference on Decision and Control
  (CDC)}.\hskip 1em plus 0.5em minus 0.4em\relax IEEE, 2016, pp. 1009--1014.

\bibitem{landau1995flexible}
I.~Landau, D.~Rey, A.~Karimi, A.~Voda, and A.~Franco, ``A flexible transmission
  system as a benchmark for robust digital control,'' \emph{European Journal of
  Control}, vol.~1, pp. 77--96, 1995.

\bibitem{bard2013practical}
J.~F. Bard, \emph{Practical bilevel optimization: algorithms and
  applications}.\hskip 1em plus 0.5em minus 0.4em\relax Springer Science \&
  Business Media, 2013, vol.~30.

\bibitem{sinha2017evolutionary}
A.~Sinha, T.~Soun, and K.~Deb, ``Evolutionary bilevel optimization using kkt
  proximity measure,'' in \emph{2017 IEEE Congress on Evolutionary Computation
  (CEC)}.\hskip 1em plus 0.5em minus 0.4em\relax IEEE, 2017, pp. 2412--2419.

\bibitem{dempe2015bilevel}
S.~Dempe, V.~Kalashnikov, G.~A. P{\'e}rez-Vald{\'e}s, and N.~Kalashnykova,
  ``Bilevel programming problems,'' \emph{Energy Systems. Springer, Berlin},
  2015.

\bibitem{dempe2019solution}
S.~Dempe and S.~Franke, ``Solution of bilevel optimization problems using the
  kkt approach,'' \emph{Optimization}, vol.~68, no.~8, pp. 1471--1489, 2019.

\bibitem{mersha2008solution}
A.~G. Mersha, \emph{Solution methods for bilevel programming problems}.\hskip
  1em plus 0.5em minus 0.4em\relax Verlag Dr. Hut, 2008.

\bibitem{walsh2001scheduling}
G.~C. Walsh and H.~Ye, ``Scheduling of networked control systems,'' \emph{IEEE
  control systems magazine}, vol.~21, no.~1, pp. 57--65, 2001.

\bibitem{van1995python}
G.~Van~Rossum and F.~L. Drake~Jr, \emph{Python reference manual}.\hskip 1em
  plus 0.5em minus 0.4em\relax Centrum voor Wiskunde en Informatica Amsterdam,
  1995.

\bibitem{behnel2011cython}
S.~Behnel, R.~Bradshaw, C.~Citro, L.~Dalcin, D.~S. Seljebotn, and K.~Smith,
  ``Cython: The best of both worlds,'' \emph{Computing in Science \&
  Engineering}, vol.~13, no.~2, pp. 31--39, 2011.

\bibitem{oliphant2006guide}
T.~E. Oliphant, \emph{A guide to NumPy}.\hskip 1em plus 0.5em minus 0.4em\relax
  Trelgol Publishing USA, 2006, vol.~1.

\bibitem{2020SciPy-NMeth}
P.~{Virtanen}, R.~{Gommers}, T.~E. {Oliphant}, M.~{Haberland}, T.~{Reddy},
  D.~{Cournapeau}, E.~{Burovski}, P.~{Peterson}, W.~{Weckesser}, J.~{Bright},
  S.~J. {van der Walt}, M.~{Brett}, J.~{Wilson}, K.~{Jarrod Millman},
  N.~{Mayorov}, A.~R.~J. {Nelson}, E.~{Jones}, R.~{Kern}, E.~{Larson},
  C.~{Carey}, {\.I}.~{Polat}, Y.~{Feng}, E.~W. {Moore}, J.~{Vand erPlas},
  D.~{Laxalde}, J.~{Perktold}, R.~{Cimrman}, I.~{Henriksen}, E.~A. {Quintero},
  C.~R. {Harris}, A.~M. {Archibald}, A.~H. {Ribeiro}, F.~{Pedregosa}, P.~{van
  Mulbregt}, and S.~.~. {Contributors}, ``{SciPy 1.0: Fundamental Algorithms
  for Scientific Computing in Python},'' \emph{Nature Methods}, 2020.

\bibitem{paszke2017automatic}
A.~Paszke, S.~Gross, S.~Chintala, G.~Chanan, E.~Yang, Z.~DeVito, Z.~Lin,
  A.~Desmaison, L.~Antiga, and A.~Lerer, ``Automatic differentiation in
  pytorch,'' in \emph{NIPS-W}, 2017.

\end{thebibliography}
